\documentclass[11pt,final]{article}
\usepackage{fullpage}

\usepackage{booktabs} 

\usepackage[ruled]{algorithm2e} 

\SetAlFnt{\small}
\SetAlCapFnt{\small}
\SetAlCapNameFnt{\small}
\SetAlCapHSkip{0pt}
\IncMargin{-\parindent}
\newcommand{\notshow}[1]{}

\usepackage{amsmath}
\usepackage{amsthm}
\usepackage{amssymb}
\usepackage{mathtools}
\usepackage{dsfont}

 \usepackage[square,longnamesfirst]{natbib}

\usepackage{amsfonts}
\usepackage{titlesec}
\usepackage{comment}
\usepackage{url}
\usepackage{verbatim}
\usepackage[shortlabels]{enumitem}
\usepackage{threeparttable}
\usepackage{bbm}
\allowdisplaybreaks
\titlelabel{\thetitle.\quad}

\usepackage{subfigure}
\usepackage[T1]{fontenc}
\usepackage{textcomp}
\usepackage{bm}
\PassOptionsToPackage{hyphens}{url}

\newcommand{\Omit}[1]{}

\usepackage{color}

\newtheorem{theorem}{Theorem}[section]
\newtheorem{lemma}[theorem]{Lemma}

\newtheorem{proposition}[theorem]{Proposition}

\newtheorem{corollary}[theorem]{Corollary}

\newtheorem{definition}{Definition}[section]

\theoremstyle{definition}

\newtheorem{defn}{Definition}[section]

\newtheorem{lem}{Lemma}[section]
\DeclareMathAlphabet{\mathcal}{OMS}{cmsy}{m}{n}

\newcommand{\ignore}[1]{}

\def\discretionaryslash{\discretionary{/}{}{/}}
\def\discretionarydot{\discretionary{.}{}{.}}
\def\discretionarycolon{\discretionary{:}{}{:}}
{\catcode`\/\active
\catcode`\.\active
\catcode`\:\active
\gdef\URLprepare{\catcode`\/\active\let/\discretionaryslash
                 \catcode`\.\active\let.\discretionarydot
                 \catcode`\:\active\let:\discretionarycolon
        \def~{\char`\~}}}%
\def\URL{\bgroup\URLprepare\realURL}%
\def\realURL#1{\tt #1\egroup}%

\newtheorem{observation}{Observation}

\newcommand{\mhrplus}{\ensuremath{\text{MHR}^+}}

\newenvironment{numberedproposition}[1]{%
\renewcommand{\thetheorem}{#1}%
\begin{proposition}}{\end{proposition}\addtocounter{theorem}{-1}}

\begin{document}
\title{When to Limit Market Entry under Mandatory Purchase}

\author{Meryem Essaidi%
\thanks{%
    {Princeton University (\protect\url{messaidi@cs.princeton.edu}).}}
\and Kira Goldner%
\thanks{%
    {Columbia University (\protect\url{kgoldner@cs.columbia.edu}).  Supported in part by NSF CCF-1420381 and by a Microsoft Research PhD Fellowship.  Supported in part by NSF award DMS-1903037 and a Columbia Data Science Institute postdoctoral fellowship.}}
\and S. Matthew Weinberg%
\thanks{%
	{Princeton University (\protect\url{smweinberg@princeton.edu}).  Supported by NSF CCF-1717899.}}
}

\maketitle



\begin{abstract}

We study a problem inspired by regulated health insurance markets, such as those created by the government in the Affordable Care Act Exchanges or by employers when they contract with private insurers to provide plans for their employees.  The market regulator can choose to do nothing, running a Free Market, or can exercise her regulatory power by limiting the entry of providers (decreasing consumer \emph{welfare} by limiting options, but also decreasing \emph{revenue} via enhanced competition). We investigate whether limiting entry increases or decreases the \emph{utility} (welfare minus revenue) of the consumers who purchase from the providers, specifically in settings where the outside option of ``purchasing nothing'' is prohibitively undesirable.

We focus primarily on the case where providers are symmetric. We propose a sufficient condition on the distribution of consumer values for (a) a unique symmetric equilibrium to exist in both markets and (b) utility to be higher with limited entry.  (We also establish that these conclusions do not necessarily hold for all distributions, and therefore some condition is necessary.) Our techniques are primarily based on tools from revenue maximization, and in particular Myerson's virtual value theory. We also consider extensions to settings where providers have identical costs for providing plans, and to two providers with an asymmetric distribution.


\end{abstract}

\section{Introduction}
\label{sec:intro}


In the Affordable Care Act (ACA) Exchanges, the government plays the role of the \emph{market regulator}, contracting with several private health insurance \emph{providers} to offer insurance plans to a number of patients, or \emph{consumers}.  Each insurance provider posts a price for his insurance plan, which he can sell to an unlimited number of consumers, and then each consumer chooses a plan to purchase among those for sale. Further, the ACA originally imposed an ``individual mandate'' that charges a steep penalty to consumers who do not purchase any plan, whose aim was to essentially make it mandatory for every consumer to purchase some plan. 




In the ACA exchanges, prices set by providers tend to be high\footnote{In 2013, even with subsidies, premiums represented between 4 and 9.5\% of the median income \citep{johnson2013can}. In 2015, a survey found that without subsidies, average marketplace unsubsidized premiums were over $2.5\times$ what most consumers claim is the maximum they can afford \citep{williams2016patient}.}, and little regulatory power exercised is by the government (for example, the ACA establishes minimum coverage guarantees, but otherwise is not an active regulator). One reason cited for these high prices is that the insurance market has a high barrier to entry, so once incumbent providers have surpassed this barrier, they face little competition. The motivating scenario to have in mind is that there may be multiple incumbents but, due to entry barriers, each incumbent serves a section of the market significantly better than all others. Could the government improve consumer utility by further regulating the market? For instance, what if the government were to only allow the insurance providers offering the five least expensive plans to enter the exchanges?  This would create competition and ensure lower prices. However, it would also mean fewer options, excluding the favorite choices of some consumers' from the market, and thus leading to lower consumer \emph{welfare}. That is, such regulation decreases both prices and welfare, so it not immediately obvious whether it improves consumer \emph{utility} (which is the difference of the two). The main purpose of this paper is to provide theoretical tools to reason about this tradeoff, and understand under what circumstances consumer utility is improved.

The ACA exchanges are an example of a regulated health insurance market. Another common example concerns an employer who contracts with providers in order to provide health insurance to her employees.  In this paper, we focus on such regulated markets, where a market regulator can  determine which providers enter the market, and where every consumer \emph{must} purchase a plan.\footnote{Put another way, the outside option of purchasing nothing is so undesirable that all rational consumers will choose to purchase something, even at negative utility. This may be due to an explicit fine (individual mandate), or simply because the outside option of being uninsured is so undesirable. But we do not explicitly model this undesirable outside option.}  Note that while healthcare is our motivating example, the model is not limited to only healthcare, but applies to any situation where purchase is mandatory.   

The focus of this paper is the following: \emph{under what conditions does limited entry (reducing prices via competition, but also reducing welfare via reduced options) improve consumer utility?}


Observe that reasoning about this question requires a direct understanding of what prices the providers will choose to set. As such, our main techniques involve characterizing and analyzing equilibria, and understanding when they exist and are unique. We compare two settings: the \emph{Free Market} setting, where the market regulator does not restrict entry at all, and the \emph{Limited Entry} setting. Note, of course, that while we use healthcare exchanges as a motivating example, the focus of our paper is to provide a model and theoretical tools which are generally applicable. As such, our model is stylized and intended to capture one key decision facing regulators of healthcare exchanges (whether to limit entry or not)---it is not intended to capture verbatim the full range of challenges facing healthcare regulators. Still, we do emphasize that this is indeed a key decision facing regulators of healthcare exchanges, and that theoretical tools from mechanism design indeed provide insight to make these decisions.



\subsection{Model and Results}



Our goal is to provide a clean model to reason about the impact of limiting entry. To this end, we consider $n$ providers, and consider a population $D$ of consumers, each with a value $v_i$ for the plan offered by provider $i$. For the majority of the paper, we consider the symmetric setting, where a random consumer from $D$ has value $v_i$ for plan $i$ drawn i.i.d. from some single-variate distribution $F$ (that is, while individual consumers have different values for different providers, the providers are a priori identical).\footnote{To have an example in mind, consider a provider which offers a good network for patients with diabetes, and another which offers a good network for patients with a history of cancer. Any particular patient will of course have different values for the plans, based on their own medical history. But the population as a whole does not necessarily prefer one over the other.}  We denote this as $D:= F^n$. 

In Section~\ref{sec:design}, we study the Free Market setting, where each provider $i$ simply sets a price $p_i$, and consumers drawn from $D$ purchase the plan $i$ which maximizes their utility $v_i - p_i$. We first identify sufficient conditions on $F$ (involving a new condition we term ``\mhrplus'') for the existence of a unique symmetric equilibrium, and we characterize the equilibrium prices in the Free Market. Here, $h_2^n(F)$ denotes the expected hazard rate of the second-highest of $n$ i.i.d. draws from $F$, and formal definitions of both technical terms appear in Section~\ref{sec:prelim}. 

\begin{theorem}\label{thm:symmetriceq} Let $D:=F^n$, where $F$ is \mhrplus\ and has decreasing density. Then the unique symmetric equilibrium in the Free Market setting is for each provider to set price $1/h_2^n(F)$.
\end{theorem}

Observe that a characterization of a canonical equilibrium in the Free Market setting is necessary if we are to possibly anlayze consumer utility. We include one vignette regarding our technical approach, which leverages machinery from the revenue maximization literature~\cite{Myerson81}. Suppose all other providers are setting price $p$; what is a provider's best-response?  To reason about this, consider instead a new distribution $F_p^*$ such that $1-F_p^*(q)$ is the probability that a consumer drawn from marginal $F$ will purchase from this provider if he sets price $q$, given that all other providers set price $p$.  The provider's best response is to set $q^*$ which maximizes his expected profit, $q\cdot (1-F_p^*(q))$, and $p$ is therefore a symmetric equilibrium if and only if $p = q^*$. Using this rephrasing, we then argue that if this distribution $F_p^*$ happens to have a \emph{monotone hazard rate} (MHR), then we are guaranteed the existence of a unique equilibrium. Of course, this distribution $F_p^*$ is quite different from $F$ itself (for example, $F$ may be MHR and $F_p^*$ may not even be regular!). We define a new distributional assumption, \mhrplus, such that if $F$ is \mhrplus then this implies that $F_p^*$ is MHR. We postpone a formal definition of \mhrplus to Section~\ref{sec:design}, but only note here that it is a strictly stronger condition than MHR, and that most common MHR distributions are also \mhrplus (e.g. exponential, uniform, Gaussian).

Next, in Section~\ref{sec:eval}, we study the Limited Entry setting. Formally, each provider still sets a price $p_i$, but now only the $n-1$ providers with lowest price enter the market (tie-breaking arbitrarily). Consumers again pick the plan maximizing $v_i - p_i$, but only among these $n-1$ providers. For symmetric instances, quickly observe that there is a unique symmetric equilibrium, and in it all providers set $p_i=0$ (so in some sense, our model can be seen as ``optimistic'' towards the benefits of limiting entry).  The main result of this section is a characterization of the precise condition on $F$ that implies that consumer utility in the Limited Entry setting will be greater than in the Free Market setting; we call this the Limit-Entry condition.

\begin{definition}[Limit-Entry Condition] We say that a symmetric distribution $D = F^n$ satisfies the Limit-Entry Condition if $H_1^n(F) \leq n/h_2^n(F)$. Here, $H_1^n(F)$ is the expected inverse hazard rate of the highest of $n$ i.i.d. draws from $F$. (Recall that $h_2^n(F)$ is the expected hazard rate of the second-highest of $n$ i.i.d. draws from $F$.)
\end{definition}

\begin{theorem}\label{thm:limit} Let $D:= F^n$ and admit a symmetric equilibrium in the Free Market setting. Then the expected consumer utility at the unique symmetric equilibrium in the Limited Entry setting exceeds that at the unique symmetric equilibrium in the Free Market setting \emph{if and only if} $F$ satisfies the Limit-Entry Condition.
\end{theorem}

Finally, while the Limit-Entry condition is relatively clean, it is not obvious how it relates to more common distributional assumptions. Our final main result shows that the Limit-Entry condition is satisfied under standard distributional assumptions.

\begin{theorem}\label{thm:mhrlimit} Let $D:=F^n$, where $F$ is MHR and has decreasing density. Then $F$ satisfies the Limit-Entry Condition.
\end{theorem}

\begin{corollary}\label{cor:main}
Let $D:=F^n$, where $F$ is \mhrplus\ and has decreasing density. Then $D$ has a unique symmetric equilibrium in the Free Market setting, and the expected consumer utility at this unique equilibrium is exceeded by the expected consumer utility in the Limited Entry setting.
\end{corollary}


In the interest of completeness, we examine whether any of our assumptions can possibly be relaxed to more standard assumptions (e.g. \mhrplus\,to MHR). In short, Proposition~\ref{prop:examples} establishes that the answer is no, suggesting that there is indeed a relevant aspect of our stronger assumptions as it relates to our conclusions.

\begin{proposition}\label{prop:examples} (Different) distributions $D:=F^n$ with the following properties all exist:
\begin{itemize}
\item $F$ is MHR, but there exists a $p$ for which $F^*_p$ is not MHR (in fact, it is anti-MHR).
\item $F$ is MHR, but $D$ has no symmetric equilibrium in the Free Market setting.
\item $D$ has a symmetric equilibrium in the Free Market setting, but does not satisfy the Limit-Entry Condition.
\item $D$ satisfies the Limit-Entry Condition, but does not have a symmetric equilibrium in the Free Market setting.
\end{itemize}
\end{proposition}

Finally, we include some extensions.  In Section~\ref{sec:costs}, we extend our results to the setting where providers have identical costs for providing service.  In Section~\ref{sec:asym}, we consider the case of asymmetric distributions $D = \times_i F_i$.  We prove that if $F_i$ is \mhrplus, then the asymmetric analog, $F^*_{i,\vec{p}_{-i}}$, is MHR.  We also show that for two providers, when $F^*_{i,\vec{p}_{-i}}$ is MHR, an equilibrium exists.

\subsection{Related Work}
\label{sec:relwork}



In order to compare with the literature on procurement auctions, we will call the providers ``suppliers.'' To the best of our knowledge, this is the first work to study procurement of suppliers under mandatory purchase of consumers.  The following literature review pertains to procurement auctions \emph{without} mandatory purchase for consumers. The most relevant work is that of \citep{saban2019procurement}.  They also study procurement auctions with $n$ heterogenous goods, each owned by a different supplier, a consumer population, and a mechanism designer whose objective is to maximize consumer utility.  The simplest comparison between these works is that~\citep{saban2019procurement} study a wide variety of different procurement models, whereas we focus on depth of one particular model. For example, our work fits into their ``First-Price Auction'' model (which is only one of many models they consider). But within this model, they consider only a two-supplier setting in a simple Hotelling game~\citep{hotelling1929},\footnote{Specifically: consumers are uniformly distributed along the unit line. Each supplier offers an item with fixed value at the endpoint of the line, and consumers value the items at their value minus the distance.} whereas we study this model in significantly more depth and generality.


 


Other work in procurement \citep{anton2004regulation,mcguire1995incomplete} also studies optimal centralized allocations for consumer utility, where the designer can choose which suppliers allocate to which consumers.  We do not study this form of allocation at all.


Most prior work studies the two-supplier case without mandatory purchase.  \citep{engel2002competition} also uses a third-party mechanism designer apart from consumers and suppliers, but studies different objectives than us: (1) welfare (consumer utility plus supplier revenue) and (2) supplier revenue minimization. They do not, however, study consumer utility maximization. 


Other works allow the consumers to act as the auctioneer \citep{chen2013group,dana2012buyer} and investigate whether it is better for the consumers to have one or two suppliers; the answer differs depending on whether the consumers' information is private or not.  These papers do not have a mechanism designer acting separately from the consumers; they also only study the stylized Hotelling model.






\subsection{Brief Summary}
We study consumer utility in a market with $n$ providers under mandatory purchase, motivated by the current state of ACA exchanges (Free Market) versus employer insurance markets (Limited Entry). We find clean sufficient conditions for equilibria to exist (Theorem~\ref{thm:symmetriceq}) in the Free Market setting, and establish that conditions like these are also necessary (Proposition~\ref{prop:examples}). We also establish clean necessary and sufficient conditions for consumer utility to improve with Limited Entry over the Free Market (Theorems~\ref{thm:limit} and~\ref{thm:mhrlimit}, and Corollary~\ref{cor:main}).

We also wish to briefly note the technical highlights. Typically, establishing existence/uniqueness of market equilibria requires solving a system of non-linear equations (and establishing uniqueness). Of course, our proofs must also accomplish this, but we get a surprising amount of leverage via Myersonian virtual value theory. That is, we interpret equilibrium conditions as one price being revenue-maximizing for a related consumer distribution. Due to mandatory purchase, this interpretation (while mathematically involved) is conceptually fairly clean. This enables us to break down a complex mathematical proof into conceptually digestible chunks, and also provides insight into the right conditions to search for. We are optimistic that these tools will continue to be useful in extensions beyond those considered explictly in this paper.
\section{Notation and Preliminaries}\label{sec:prelim}

We consider the following problem from the perspective of a market regulator. We use the language of healthcare providers throughout the paper (although we again remind the reader that healthcare exchanges are just a motivating example for our stylized model). There are $n$ providers, each of whom produces a single (distinct) plan. Each individual consumer in the market has a valuation vector $\vec{v}\in \mathbb{R}^n_+$ for the plans, with $v_i$ denoting their value for plan $i$. The market consists of a continuum over valuations $\vec{v}$, which can alternatively be interpreted as a distribution $D$ (over a random consumer drawn from the market).

We assume throughout the paper that $D$ is a product distribution (that is, $D:= \times_i D_i$ for single-dimensional $D_i$). We will use $F_i$ to denote the CDF of $D_i$, and assume that each $D_i$ also has a density function, or PDF, denoted by $f_i$. For our main results, we will also assume that $D$ is symmetric (that is, $D_i = D_j$ for all $i,j$, or the valuations are identically drawn across providers). In Section~\ref{sec:asym}, we consider extensions to asymmetric distributions.

In our context, let's briefly elaborate on these assumptions. Assuming that each $D_i$ admits a density function is extremely common in past literature (e.g.~\cite{Myerson81}), and is comparable to a ``large market'' assumption that no particular individual has an oversized role. The motivation for this assumption is purely technical, since it allows for clean closed-form definitions of conditions such as regularity or Monotone Hazard Rate. Assuming that $D$ is a product distribution is also extremely common (e.g.~\cite{Myerson81,ChawlaHK07, ChawlaHMS10}), and corresponds to the property that a consumer's value for one plan does not influence the probability of their value for another. While this assumption may initially appear restrictive, numerous works establish that results proved in this setting generally extend to richer settings as well. Indeed, our results immediately extend, for free, to the ``common base-value'' model of~\cite{ChawlaMS15},\footnote{In the common base-value model, each consumer has a ``base value'' for all plans, and an idiosynchratic value for each plan separately. Their value for a plan sums these two together.} but we focus on the independent setting for ease of exposition (see Section~\ref{sec:costs} for details on this particular extension). Assuming that $D$ is symmetric corresponds to the following: \emph{individuals} may certainly have distinct values for distinct plans. The fact that $D$ is symmetric simply means that \emph{a priori} there is nothing special about one plan versus another.   \\

\noindent{\textbf{Free Market Setting:}} In the free market setting, each provider $i$ sets a price $p_i$ on their plan. A consumer drawn from $D$ purchases the plan $i^* = \arg\max_{i}\{v_i - p_i\}$. Importantly, notice that the consumer \emph{must} purchase a plan, even if $v_i < p_i$ for all $i$. So provider $i$'s payoff is equal to $p_i \cdot \Pr_{\vec{v} \leftarrow D}[i = \arg\max_j \{v_j - p_j\}]$.\footnote{Observe also that because each $D_i$ has a PDF, ties occur with probability $0$, and there is always a unique $\arg\max$. As a result, we will not be careful between $\leq$ and $<$ when discussing preferences.} A best response of provider $i$ to $\vec{p}_{-i}$, where $-i$ denotes all agents other than $i$, is the payoff-maximizing price in response to $\vec{p}_{-i}$. A price vector $\vec{p}$ is a pure equilibrium if each provider is simultaneously best responding. An equilibrium $\vec{p}$ is symmetric if $p_i = p_j$ for all $i, j$. Observe that when both $D$ and $\vec{p}$ are symmetric, the payoff to each provider is just $p_i/n$. \\

\noindent{\textbf{Limited Entry Setting:}} The focus of this paper is contrasting the Free Market setting with a ``Limited Entry'' setting. In the Limited Entry setting, the regulator does not exert total control over the market (e.g. by directly setting prices), but simply restricts entry to the market. Specifically, each provider $i$ first sets a price $p_i$ on their plan, and then the market regulator selects a subset $S$ of $k < n$ providers to enter according to a pre-specified rule based only on the relative ordering of the $p_i$s (note that we require the regulator to always kick out at least one provider, so that the Free Market setting is not a special case of Limited Entry).\footnote{If the rule can depend on the prices themselves, then the regulator can effectively set the market prices directly by kicking everyone out of the market unless they set the desired prices.} A consumer drawn from $D$ purchases the plan $i^* = \arg\max_{i \in S}\{v_i - p_i\}$. Again, the consumer \emph{must} purchase a plan in $S$, even if $v_i < p_i$ for all $i \in S$. Provider $i$'s payoff is $0$ if they are not selected to be in $S$, or equal to $p_i \cdot \Pr_{\vec{v} \leftarrow D} [i = \arg\max_{j \in S} \{v_j - p_j\}]$ otherwise. A price vector $\vec{p}$ is again an equilibrium if each provider is simultaneously best responding, and symmetric if $p_i = p_j$ for all $i, j$. 

It is not hard to see when $D$ is symmetric that among all selection rules, and all equilibria of those rules, the one which results in highest consumer utility is to take the $k=n-1$ providers who set the lowest prices (tie-breaking randomly
), due to the following observation. As a result, we will simply refer to this particular rule as ``the'' Limited Entry setting, and study only this rule for the rest of the paper.

\begin{observation}\label{obs:limited}
For all symmetric $D$, the unique equilibrium in the Limited Entry setting is $\vec{p} = \vec{0}$.  
\end{observation}

This is because the losing provider earns no profit and is better off undercutting.  The same is true for the next losing provider, and so on, until the prices reach $\vec{0}$. \\
\\
\noindent{\textbf{Consumer Utility:}} The focus of this paper is understanding the expected consumer \emph{utility} in equilibrium for both settings. Specifically, the expected consumer \emph{welfare} is equal to $\mathbb{E}_{\vec{v}\leftarrow D}[v_{i^*},\ \ i^* = \arg\max_{i}\{v_i-p_i\}]$, where the argmax is taken over $i \in [n]$ in the free market setting, or $i \in S$ in the Limited Entry setting. That is, the expected welfare is simply the expected value a consumer receives for the plan they purchase. The expected \emph{revenue} is simply the sum of the providers' payoffs. Consumer utility is then just welfare minus revenue.  Recall again that consumers can \emph{not} opt out, and some consumers may indeed get negative utility. 

\subsection{Distributional Properties}
Symmetric equilibria do not always exist for symmetric distributions (Proposition~\ref{prop:examples}), and limiting entry does not universally increase or decrease consumer utility compared to the Free Market setting (also Proposition~\ref{prop:examples}). As such, the focus of this paper is in providing sufficient conditions for (e.g.) (1) equilibria to exist, and (2) limiting entry to improve consumer utility. Below are properties of single-variate distributions which we'll use. The first two are standard in the literature. \mhrplus\ is a new condition we introduce which is a proper subset of MHR (see Observation~\ref{obs:defs}), and happens to be ``the right'' restriction of MHR for our setting. For all definitions below, ``non-decreasing'' means ``non-decreasing over the support of $F$,'' and ``for all $x$'' means ``for all $x$ in the support of $F$.''
Modulo our new \mhrplus, each of these conditions are common assumptions in past work (e.g.~\cite{BulowK96,PaiV10}). In typical applications, regularity (or at least MHR) suffice for desired positive results to hold. Proposition~\ref{prop:examples} establishes, perhaps surprisingly, that MHR doesn't suffice in our setting, motivating the \mhrplus\ definition.

\begin{definition}[Regular] A one-dimensional distribution with CDF $F$ and PDF $f$ is \emph{regular} if for all $x$, $x - \frac{1-F(x)}{f(x)}$ is monotone non-decreasing.
\end{definition}

\begin{definition}[Monotone Hazard Rate (MHR)] A one-dimensional distribution with CDF $F$ and PDF $f$ is \emph{MHR} if for all $x$, the hazard rate $h_F(x) :=\frac{f(x)}{1-F(x)}$ is monotone non-decreasing. 
\end{definition}


The following condition is new to this work.  Note that $f'(x)$ denotes $\frac{d}{dx} f(x)$.  The same is true for the notation $'$ throughout.  

\begin{definition}[\mhrplus] A one-dimensional distribution with CDF $F$ and PDF $f$ is \emph{\mhrplus} if there exists a constant $c \geq 0$ such that $cf(x) \geq -f'(x)$ and $h_F(x) \geq c$ for all $x$.
\end{definition}

\begin{definition}[Decreasing Density] A one-dimensional distribution with CDF $F$ and PDF $f$ has \emph{decreasing density} if $f(\cdot)$ is non-increasing.
\end{definition}

The following observation provides several equivalent conditions for the above definitions. In particular, the second condition concerning MHR (\ref{dist4}) and the second condition concerning \mhrplus (\ref{dist6}) establish how \mhrplus distributions are MHR distributions ``plus a little extra.''

\begin{observation}\label{obs:defs}The definitions above are equivalent to the following conditions: 
\begin{enumerate}
\item A distribution is regular iff $2f(x)^2 \geq -f'(x) (1-F(x))$ for all $x$. \label{dist1}
\item A distribution is regular iff $2f(x) h_F(x) \geq -f'(x)$ for all $x$. \label{dist2}
\item A distribution is MHR iff $f(x)^2 \geq -f'(x) (1-F(x))$ for all $x$. \label{dist3}
\item A distribution is MHR iff $f(x) h_F(x) \geq -f'(x)$ for all $x$. \label{dist4}
\item A distribution is \mhrplus\ iff it is MHR and $f(x) f(0) \geq -f'(x)$ for all $x$. \label{dist5}
\item A distribution is \mhrplus\ iff it is MHR and $f(x) h_F(0) \geq -f'(x)$ for all $x$. \label{dist6}
\item A distribution is \mhrplus\ iff $f(x) f(0) \geq -f'(x)$ and $h_F(x) \geq f(0)$ for all $x$. \label{dist7}
\end{enumerate}
\end{observation}
\begin{proof}
The first four conditions follow immediately from taking the derivative of a non-decreasing function, and ensuring that it is $\geq 0$ (and plugging in the definition of $h_F(\cdot)$). The fifth/sixth conditions are equivalent as $1-F(0) = 1$, so we just need to confirm that they follow from the definition of \mhrplus\ and the earlier conditions. Indeed, if a distribution is \mhrplus, then it is clearly MHR (by (\ref{dist4})). Also, as $h_F(0) \geq c$, and $c f(x) \geq -f'(x)$ for all $x$, then $f(x) h_F(0) \geq -f'(x)$ for all $x$. Similarly, if a distribution satisfies (\ref{dist6}), define $c:= h_F(0) = f(0)$. Then clearly $cf(x) \geq -f'(x)$ immediately from (\ref{dist6}), and also because the distribution is MHR, we have $h_F(x) \geq h_F(0) = c$ for all $x$, so the distribution is \mhrplus.

Finally to see that (\ref{dist7}) implies \mhrplus, observe that $f(0) \geq 0$ is the desired $c$ for the definition of \mhrplus. Also, any distribution that satisfies (\ref{dist6}) clearly has $f(x) f(0) \geq -f'(x)$ (as $f(0) = h_F(0)$). And, as distributions satisfying (\ref{dist6}) are MHR, we also have $h_F(x) \geq h_F(0) = f(0)$, as desired.
\end{proof}

We will use condition (\ref{dist3}) for MHR and (\ref{dist7}) for \mhrplus several times throughout the proofs in Sections~\ref{sec:design} and \ref{sec:asym}.  See Figure~\ref{fig:distclasses} for examples of distributions in each class.

\begin{figure}[h!]
\centering
\includegraphics[scale=.9]{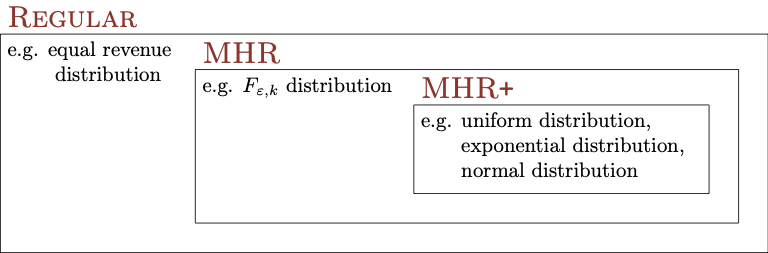}  
\caption{Recall that $\mhrplus \subset \text{MHR} \subset \text{regular}$.  The above shows examples of regular-only, MHR-but-not-\mhrplus, and \mhrplus distributions.  Note that essentially every MHR distribution is also \mhrplus except for those explicitly constructed so as not to be.  The $F_{\varepsilon,k}$ distribution can be found in Appendix~\ref{app:examples}.}
\label{fig:distclasses}
\end{figure}

Finally, we'll use the following notation for many of our theorem statements. 

\begin{definition}[Expected Order Statistics]
For a single-variate distribution with CDF $F$, define:
\begin{itemize}
\item $X_i^n(F)$ to be a random variable which is the $i^{th}$ highest of $n$ i.i.d. draws from $F$.
\item $V_i^n(F)$ to be the expected value of the $i^{th}$ highest of $n$ i.i.d. draws from $F$. That is, $V_i^n(F) := \mathbb{E}[X_i^n(F)]$.
\item $h_i^n(F)$ to be the expected hazard rate of the $i^{th}$ highest of $n$ i.i.d. draws from $F$. That is, $h_i^n(F) := \mathbb{E}\left[\frac{f(X_i^n(F))}{1-F(X_i^n(F))}\right]$. Note that the definition inside the expectation is intentional: we first find the $i^{th}$ highest sample, and then compute its hazard rate with respect to the original $F,f$.  
\item $H_i^n(F)$ to be the expected inverse hazard rate of the $i^{th}$ highest of $n$ i.i.d. draws from $F$. That is, $H_i^n(F) := \mathbb{E}\left[\frac{1-F(X_i^n(F))}{f(X_i^n(F))}\right]$.
\end{itemize}
\end{definition}

\subsection{Virtual Value Preliminaries}
A tool that we'll repeatedly use throughout our results is the Myersonian theory of virtual values~\cite{Myerson81}.

\begin{definition}[Virtual value] For a single-dimensional distribution with CDF $F$ and PDF $f$, define the \emph{virtual value} with respect to $F$ as $\varphi_F(\cdot)$, with $\varphi_F(v):= v - \frac{1-F(v)}{f(v)}$.
\end{definition}

Note also that $\varphi_F(v) = v - \frac{1}{h_F(v)}$. The inverse of $\varphi_F(v)$ is well-defined when $\varphi_F(v)$ is non-decreasing, and unique when strictly increasing.

\begin{theorem}[\cite{Myerson81}]\label{thm:myerson} The following conditions hold
\begin{itemize}
\item Let $F$ be regular. Then $\varphi_F(\cdot)$ is monotone non-decreasing, and $\arg\max_p\{p \cdot (1-F(p))\} = \varphi_F^{-1}(0)$. Observe that $\varphi_F^{-1}(0)$ is the set of all $v$ such that $v = 1/h_F(v)$. 

\item If $F$ is not regular, it is still the case that $q:= \arg\max_p \{p \cdot (1-F(p))\}$ satisfies $\varphi_F(q) = 0$.\footnote{For readers not familiar with this particular claim, it follows by observing that $\varphi_F(q) = 0$ is exactly the first-order condition for maximizing the revenue curve of $F$.}

\item Finally, for all $n$ and $F$, $\mathbb{E}[X_2^n(F)] = \mathbb{E}[\varphi_F(X_1^n(F))]$.\footnote{This follows from the equivalence of expected virtual surplus and expected revenue. The expected revenue of the second-price auction with $n$ bidders is the LHS, and the expected virtual value of the winner is the RHS.}
\end{itemize}
\end{theorem}

\begin{observation}Let $F$ be MHR. Then $\arg\max_p\{p \cdot (1-F(p))\}$ is unique. 
\end{observation}
\begin{proof}
Observe that $v$ is strictly increasing in $v$, and $1/h_F(v)$ is weakly decreasing in $v$. Therefore, $v = 1/h_F(v)$ cannot have multiple solutions.
\end{proof}

\section{Best Responding in the Free Market Setting}
\label{sec:design}
In this section, we expand on the mathematics behind what it means to best-respond in the Free Market setting. Importantly, recall that every consumer \emph{must} select a provider, even if their utility for each is negative. Our focus is on the symmetric case ($D$ is symmetric, searching for a symmetric equilibrium). In Section~\ref{sec:asym} we consider extensions to the asymmetric case.

Consider the search for a symmetric equilibrium. The question we need to ask is ``given that the $n-1$ other providers are setting price $p$, is $p$ a best response for the remaining provider?'' To resolve this, we first need to understand more precisely the payoff received by the remaining provider for setting price $q$ while the other $n-1$ set price $p$.

\begin{defn}[Star Operation] Let $F^*_p(q)$ be such that when all providers $j \neq i$ are setting price $p$, and provider $i$ sets price $q$, the probability that the consumer purchases from provider $i$ is $1-F^*_p(q)$. That is, the payoff to provider $i$ in this circumstance is $q(1-F^*_p(q))$. 
\end{defn}

Our choice of notation suggests that we will reason about best-responding as a single-item revenue problem, with the consumer's value distribution defined by $F^*_p$. Our goal will be to find sufficient conditions for there to exist a $p$ such that $p$ itself is the revenue-maximizing price for the distribution $F^*_p$. Our plan is roughly as follows:
\begin{itemize}
\item Write an expression for $\varphi_{F^*_p}(\cdot)$. 
\item Observe that $\varphi_{F^*_p}(p) = 0$ is \emph{necessary} for $p$ to possibly be a symmetric equilibrium, by first-order conditions (Theorem~\ref{thm:myerson}). Show that this equation has a unique solution.
\item If $F^*_p$ is regular or MHR, then first-order conditions suffice for $p = \varphi^{-1}_{F^*_p}(0)$ to be a best response, and therefore $p$ is indeed a symmetric equilibrium (but this is not \emph{necessary} for $p$ to be a symmetric equilibrium).
\item Prove that if $F$ is \mhrplus\ with decreasing density, then $F^*_p$ is MHR.
\end{itemize}

Let's quickly highlight some aspects of this plan. Typically, finding a closed form for potential equilibria, and establishing that sufficient conditions hold is a matter of solving systems of non-linear equations. Often, this process may be mathematically engaging, but not offer insight connecting the sufficient conditions to the conclusions. The final step of our outline (that \mhrplus\ $F$ implies MHR $F^*_p$) is still mostly ``just math;'' however, the rest of the outline leverages existing theory of Myersonian virtual values to make the rest of the math more intuitive.

Let's now begin by writing an analytical expression for $F^*_p(q)$, $f^*_p(q)$, and $(f^*)'_p(q)$. This will let us (a) compute the virtual value $\varphi_{F^*_p}(q)$ and (b) check whether $F^*_p$ is regular or MHR. Recall that $1-F^*_p(q)$ is the probability that the consumer chooses to purchase from a provider priced at $q$ when the other $n-1$ are priced at $p$.  Below, observe that $g_p(q,x)$ is the density of the maximum of $n-1$ draws from $F$, after adding $p$ and subtracting $q$.  

\begin{proposition}\label{prop:stars}Let $g_p(q,x):= (n-1)f(x-q+p)(F(x-q+p))^{n-2}$. Let also $M:= \max\{0, q-p\}$.
\begin{itemize}
\item $F^*_p(q) = \int_M^\infty F(x) g_p(q,x)dx$.
\item $1-F^*_p(q) = \int_M^\infty (1-F(x)) g_p(q,x)dx+F(M+p-q)^{n-1}$.
\item $f^*_p(q) = \int_M^\infty f(x) g_p(q,x)dx$.
\item $(f^*_p)'(q) = \int_M^\infty (f'(x)) g_p(q,x) dx+f(0)g_p(q,M)$.
\end{itemize}
\end{proposition}


Note that the definitions in Proposition~\ref{prop:stars} are referred to many times throughout the proofs of Propositions~\ref{prop:symmetric} and \ref{prop:mainstars}; the reader may want to keep them handy.  The proof of Proposition~\ref{prop:stars} appears in Appendix~\ref{app:proofs}.  

The following describes a condition that must be met as a result of first-order conditions in order to possibly have a symmetric equilibrium in the Free Market setting.  Note that this holds for \emph{any} $F$, even those which are not MHR or regular.

\begin{proposition}\label{prop:symmetric}
Let $D:= F^n$. The only possible symmetric equilibrium in the Free Market setting is $p_F:= \frac{1}{h_2^n(F)}$. If $F^*_{p_F}$ is regular, then $\frac{1}{h_2^n(F)}$ is a symmetric equilibrium.
\end{proposition}
\begin{proof}
By Theorem~\ref{thm:myerson} and the definition of $F^*_p$, in order for $p$ to be a best response to all other providers setting price $p$, we must have $\varphi_{F^*_p}(p) = 0$. Note that this does not guarantee that $p$ is a best response; this is just a necessary first-order condition.

Observe that, using Proposition \ref{prop:stars}, many of the terms in $F^*_p(q)$ (and $f^*_p(q)$) simplify when $p= q$, so we get:
\begin{align*}
\varphi_{F^*_p}(p) &= p - \frac{1-F^*_p(p)}{f^*_p(p)}\\
&= p - \frac{\int_0^\infty (1-F(x))(n-1)f(x) F(x)^{n-2}dx}{\int_0^\infty f(x) (n-1)f(x) F(x)^{n-2}dx}.
\end{align*}

Let's first examine the numerator. The numerator integrates over all $x$, the density $f(x)$, times the probability that exactly one of $n-1$ other draws from $F$ exceed $x$ (this is $(n-1)(1-F(x))F(x)^{n-2}$). This is exactly the probability that one of $n$ draws is the second-highest, which is just $1/n$. Another way to see that $1-F^*_p(p) = 1/n$ is just that $1-F^*_p(p)$, by definition, is the probability that a particular one of $n$ providers is the consumer's favorite. But as $D$ is symmetric, and the price $p$ is the same for all providers, this is just $1/n$. So now we wish to examine the denominator, with an extra factor of $n$ from the numerator:

$$\int_0^\infty n(n-1) f(x)^2 F(x)^{n-2} dx = \int_0^\infty n(n-1) h_F(x) f(x) F(x)^{n-2} (1-F(x))dx.$$

All we have done above is multiply and divide by $1-F(x)$. But now the integral is interpretable: we are integrating over all $x$, the number of ways to choose an ordered pair $(a,b)$ of $n$ draws ($n(n-1)$), times the density of $v_a$ at $x$ (this is $f(x)$), times the probability that $v_b$ exceeds $x$ ($1-F(x)$), times the probability that the remaining $n-2$ items do not exceed $x$ ($F(x)^{n-2}$), times the hazard rate at $x$. This is \emph{exactly} computing the expected hazard rate of the second-highest of $n$ draws from $F$! Therefore, we immediately conclude that:

$$\varphi_{F^*_p}(p) = p - \frac{1}{h_2^n(F)},$$

and therefore $\varphi_{F^*_p}(p) = 0$ iff $p = 1/h_2^n(F)$. Importantly, note that we have proven that \emph{for all $p$}, even those which are not equilibria, or otherwise related to $F$, that $\varphi_{F^*_p}(p) = p - \frac{1}{h_2^n(F)}$. 
\end{proof}

So at this point, we know know the unique candidate for a symmetric equilibrium (because it is the only candidate which satisfies first-order conditions of Theorem~\ref{thm:myerson}). If we can find sufficient conditions for $F^*_p$ to be regular (or MHR), then these first-order conditions suffice for $1/h_2^n(F)$ to indeed be a symmetric equilibrium. We identify such sufficient conditions below (remember that Proposition~\ref{prop:examples} establishes that MHR alone does not suffice, so some stronger conditions are necessary):

\begin{proposition}\label{prop:mainstars} Let $F$ be \mhrplus\ and have decreasing density. Then for all $p$, $F^*_p$ is MHR.
\end{proposition}

\begin{proof}
Lets first develop intuition for why \mhrplus\ is a convenient condition for reasoning about the starred distribution. Observe that each of the starred CDF/PDF/etc. functions are convolutions of the original CDF/PDF/etc. with $g_p(q,x)$. Unfortunately, just knowing that, for example, $f'(x)(1-F(x)) \leq f(x)^2$ for all $x$ (which is guaranteed by MHR from Observation \ref{obs:defs} (\ref{dist3})) is not enough for us to reason about these convolutions. But \mhrplus\ buys us something stronger, which is exactly what's needed for the first half of the proof: not only is $-f'(x)(1-F(x)) \leq f(x)^2$ for all $x$, but in fact $-f'(x) \leq cf(x)$ everywhere, which does allow us to make direct substitutions into the convolution.   To see this, consider attempting to manipulate $-(f^*_p)'(q)$ using only Obs.~\ref{obs:defs} (\ref{dist4}) for MHR---that $f(x)h_F(x) \geq -f'(x)$ for all $x$---instead of what we use from the definition for \mhrplus.  
The second step of the proof is dealing with the additional terms outside of the integral. Surprisingly, \mhrplus\ also turns out to be the right condition to reason transparently about these additional terms, although more creativity is required here than in step one. We now proceed with the proof.


First, observe that if $F$ has decreasing density, then $f(0) > 0$ (so we may divide by $f(0)$). Next, because $F$ is \mhrplus, recall by Observation~\ref{obs:defs} condition (\ref{dist7}) that we have $f(x) f(0) \geq -f'(x)$, and $h_F(x) \geq f(0)$.  Since $h_F(x) = f(x)/(1-F(x))$, this implies that $f(x) \geq f(0) (1-F(x))$. Therefore, we get:

\begin{align*}
1-F^*_p(q) &= \int_M^\infty (1-F(x)) g_p(q,x) dx + F(M + p-q)^{n-1} && \text{Prop. \ref{prop:stars}: Def. of $1-F^*_p(q)$}\\
&\leq \int_M^\infty \frac{f(x)}{f(0)} g_p(q,x) dx + F(M+p-q)^{n-1} && \text{Obs. \ref{obs:defs} (\ref{dist7}): $f(x) \geq f(0) (1-F(x))$} \\
&= f^*_p(q)/f(0) + F(M+p-q)^{n-1}. && \text{Prop. \ref{prop:stars}: Def. of $f^*_p(q)$}
\end{align*}

Similarly, we can write:

\begin{align*}
-(f^*_p)'(q) &= -\int_M^\infty f'(x) g_p(q,x) dx - f(0)g_p(q,M)    && \text{Prop. \ref{prop:stars}: Def. of $(f^*_p)'(q)$}\\
&\leq \int_M^\infty f(x)f(0)g_p(q,x)dx - f(0) g_p(q,M)   && \text{Obs. \ref{obs:defs} (\ref{dist7}): $f(x) f(0) \geq -f'(x)$}\\
&\leq f(0) f^*_p(q) - f(0) g_p(q,M). && \text{Prop. \ref{prop:stars}: Def. of $f^*_p(q)$}
\end{align*}

Therefore, we get:

$$(1-F^*_p(q)) \cdot (-f^*_p)'(q) \leq (f^*_p(q))^2 -f(0)g_p(q,M)\cdot (1-F^*_p(q)) -F(M+p-q)^{n-1}\int_M^\infty f'(x) g_p(q,x)dx.$$

The above inequality completes the ``step one'' referenced in the proof summary.  Our remaining task is to show that:

$$-f(0)g_p(q,M)\cdot (1-F^*_p(q)) -F(M+p-q)^{n-1}\int_M^\infty f'(x) g_p(q,x)dx \leq 0,$$

as then we will have established that $(1-F^*_p(q)) \cdot (-f^*_p)'(q) \leq (f^*_p(q))^2$, and therefore $F^*_p$ is MHR. Observe that this is clearly true when $M=q-p$, as the entire term above is $0$---both $g_p(q,q-p) = 0$ (as can be seen by plugging it in to its definition or reasoning about the probability of ties at 0, which has probability 0) and $F(0) = 0$. So the remaining case is when $M =0$. Here, we derive the following, starting with the justification and then showing the equations.

The first inequality follows by definition of $F$ being \mhrplus from Observation \ref{obs:defs} (\ref{dist7}): $f(x) f(0) \geq -f'(x)$. The second follows by dividing both sides by $f(0)g_p(g,0)$, which is strictly positive. The third line follows by evaluating the definition of $g_p(q,x)$ from Proposition~\ref{prop:stars}.

The fourth line then follows by multiplying both sides by $F^{n-1}(p-q)$, which is positive. The fifth line follows because when $M = 0$, $p > q$ and $F$ has decreasing density. The penultimate line follows as $F(p-q) \leq F(x+p-q)$ for all $x \geq 0$. 

The final line follows from the following reasoning. Recall that $1-F_p^*(q)$ denotes the probability that the consumer will choose to purchase a specific plan when that plan has price $q$ and all other plans have price $p$. The probability that this occurs conditioned on having value $x$ for the specific plan is the probability that the consumer has utility $v-p \leq x-q$ for every other plan, or value $v \leq x + p-q$, which occurs with probability exactly $F^{n-1}(x+p-q)$, and the previous term simply integrates this times $f(x)$ over all $x$.

\begin{align*}
-\int_0^\infty f'(x) g_p(q,x)dx &\leq \int_0^\infty f(0)f(x) g_p(q,x) dx\\
\Rightarrow \frac{-\int_0^\infty f'(x) g_p(q,x)dx}{f(0)g_p(q,0)}&\leq \int_0^\infty f(x)g_p(q,x)/g_p(q,0)dx\\
&= \int_0^\infty f(x) \frac{(n-1)f(x+p-q)F^{n-2}(x+p-q)}{(n-1)f(x)F^{n-2}(p-q)}dx\\
\Rightarrow  \frac{-F^{n-1}(p-q)\int_0^\infty f'(x) g_p(q,x)dx}{f(0)g_p(q,0)} & \leq F(p-q)\int_0^\infty f(x) \frac{f(x+p-q)F^{n-2}(x+p-q)}{f(x)}dx\\
&\leq F(p-q)\int_0^\infty f(x) F^{n-2}(x+p-q)dx\\
&\leq \int_0^\infty f(x) F^{n-1}(x+p-q)dx\\
&= 1-F_p^*(q).
\end{align*}

Finally, observe that this inequality is exactly what we want, as:
\begin{align*}
\frac{-F^{n-1}(p-q)\int_0^\infty f'(x) g_p(q,x)dx}{f(0)g_p(q,0)} &\leq 1-F^*_p(q) \\
\Rightarrow -F^{n-1}(p-q)\int_0^\infty f'(x) g_p(q,x)dx &\leq f(0)g_p(q,M)\cdot (1-F^*_p(q))\\
\Rightarrow -f(0)g_p(q,M)\cdot (1-F^*_p(q)) -F(M+p-q)^{n-1}\int_M^\infty f'(x) g_p(q,x)dx &\leq 0
\end{align*}
because again, $M=0$ in this case.
\end{proof}

And now, we can wrap up the proof of Theorem~\ref{thm:symmetriceq}, which claims that whenever $F$ is \mhrplus, the unique symmetric equilibrium in the Free Market setting for $D:= F^n$ is for each provider to set price $1/h_2^n(F)$.

\begin{proof}[Proof of Theorem~\ref{thm:symmetriceq}]
Proposition~\ref{prop:symmetric} establishes that $1/h_2^n(F)$ is a symmetric equilibrium as long as $F^*_{1/h_2^n(F)}$ is regular. Proposition~\ref{prop:mainstars} proves something even stronger: that $F^*_p$ is MHR for all $p$, as long as $F$ is \mhrplus\ with decreasing density. The two propositions together complete the proof.
\end{proof}

Note that Theorem~\ref{thm:symmetriceq} accomplishes several tasks:
\begin{itemize}
\item It establishes that a symmetric equilibrium exists subject to \mhrplus\ (which is not generally true without some assumptions, Proposition~\ref{prop:examples}).
\item It provides a clean closed form for this symmetric equilibrium.
\item It establishes uniqueness of this equilibrium (even stronger: this is the only possible equilibrium for all $F$). This is important because it lets us reason about ``\emph{the} utility in the Free Market setting'' without needing to worry about exactly which equilibrium we should be analyzing.
\end{itemize}

\section{Comparing Consumer Utilities}
\label{sec:eval}
In this section, we derive a Limit-Entry condition, which dictates when consumer utility is higher in the Limited Entry setting versus the Free Market. Note that our condition is well-defined even when no symmetric equilibrium exists in the Free Market setting. Let's first recall the Limit-Entry condition from Section~\ref{sec:intro}, which a distribution satisfies when $$H_1^n(F) \leq n/h_2^n(F),$$ where again, $H_1^n(F)$ is the expected inverse hazard rate of the highest of $n$ i.i.d. draws from $F$ and $h_2^n(F)$ is the expected hazard rate of the second-highest of $n$ i.i.d. draws from $F$. Recall that Theorem~\ref{thm:limit} states that consumer utility is higher in the Limited Entry setting versus the Free Market setting \emph{if and only if} the Limit-Entry condition holds. The main result of this section is a proof of Theorems~\ref{thm:limit}, and~\ref{thm:mhrlimit}.

\begin{proof}[Proof of Theorem~\ref{thm:limit}]
Let's first compute the expected consumer utility in the Limited Entry setting. 

\begin{lemma}The expected consumer utility at the unique equilibrium in the Limited Entry setting is $V_1^{n-1}(F)$.
\end{lemma}
\begin{proof}There are a total of $n-1$ providers, and recall from Observation~\ref{obs:limited} that the unique equilibrium has all prices set to $0$. Therefore, the consumer's expected payment is zero. The consumer picks their favorite plan, with value simply the maximum of $n-1$ i.i.d. draws from $F$. Together, we see that the consumer's expected utility at the unique symmetric equilibrium of the Limited Entry setting is $V_1^{n-1}(F)$.
\end{proof}

Now, let's compute the expected consumer utility in the Free Market setting. 

\begin{lemma} The expected consumer utility at the unique symmetric equilibrium (when it exists) in the Free Market setting is $V_1^n(F) - 1/h_2^n(F)$.
\end{lemma}
\begin{proof}
There are a total of $n$ providers, and the unique symmetric equilibrium (when it exists) has all prices set to $1/h_2^n(F)$. Therefore, the consumers expected payment is $1/h_2^n(F)$ (because the consumer must purchase a plan, even with negative utility for everything). The consumer's value for their favorite plan is the maximum of $n$ i.i.d. draws from $F$. Therefore, the consumer's expected utility at the unique symmetric equilibrium of the Free Market setting is $V_1^{n}(F) - 1/h_2^n(F)$.
\end{proof}

We therefore see that the expected utility is higher in the Limited Entry setting versus Free Market if and only if $V_1^{n-1}(F) \geq V_1^n(F) - 1/h_2^n(F)$. The remainder of the proof is rewriting this condition, using Myersonian virtual value theory in yet another way. We produce the steps below, and justify each step afterwards. Two of the three steps follow from basic algebra or a coupling argument. The middle step (line three) makes use of virtual value theory.

\begin{align*}
V_1^{n-1}(F) &\geq V_1^n(F) - \frac{1}{h_2^n(F)}\\
\Leftrightarrow \frac{n-1}{n}V_1^n(F)+ \frac{1}{n} V_2^n(F) &\geq V_1^n - \frac{1}{h_2^n(F)}\\
\Leftrightarrow \frac{n-1}{n}V_1^n(F) + \frac{1}{n} \left(V_1^n(F) - H_1^n(F)\right) &\geq V_1^n(F) - \frac{1}{h_2^n(F)}\\
\Leftrightarrow H_1^n(F) &\leq \frac{n}{h_2^n(F)}.
\end{align*}

The first equivalence follows by a coupling argument. One way to draw the highest of $n-1$ draws from $F$, or $X_2^{n-1}(F)$, is to take $n$ draws from $F$, remove one uniformly at random, and then examine the highest remaining.  With probability $1/n$, the highest of the $n$ draws is excluded, so the highest remaining is $X_2^n(F)$.  The rest of the time, a different draw is excluded and the highest of $n$ remains, giving $X_1^n(F)$.  Hence in expectation, $V_1^{n-1}(F) = \frac{n-1}{n}V_1^n(F) + \frac{1}{n} V_2^n(F) $.



The second equivalence follows from Theorem~\ref{thm:myerson}, as $V_2^{n}(F) = \mathbb{E}[\varphi_F(X_1^n(F))]$.  More familiarly, this is the fact that a second-price auction is revenue-maximizing, and that the revenue is equal to the virtual welfare of the highest-valued bidder in the iid setting.  Recall that $\varphi_F(v) = v - 1/h_F(v)$; then $\mathbb{E}[\varphi_F(X_1^n(F))] = V_1^n(F) - H_1^n(F)$.

The final equivalence follows by subtracting $V_1^n(F)$ from both sides and multiplying by $-1$.
\end{proof}

Finally, we prove Theorem~\ref{thm:mhrlimit}. Recall that Theorem~\ref{thm:mhrlimit} asserts that whenever $F$ is MHR with decreasing density, it satisfies the Limit-Entry condition. Recall that MHR alone is not enough to guarantee that there is an equilibrium in the Free Market setting for the Limited Entry setting to improve over, but that the condition is well-defined anyway. When $F$ is further \mhrplus, there is an equilibrium in both settings, and the consumer utility is always higher with Limited Entry (Corollary~\ref{cor:main}).

\begin{proof}[Proof of Theorem~\ref{thm:mhrlimit}]
The proof will follow from the steps below (justification for each step is provided afterwards). If $F$ is MHR with decreasing density, then:
\begin{align*}
f(0) &\geq \mathbb{E}[f(X_1^{n-1})]\\
\Leftrightarrow \frac{1}{f(0)} &\leq \frac{1}{\mathbb{E}[f(X_1^{n-1})]}\\
\Rightarrow \frac{1}{h_F(0)} & \leq \frac{n}{n\mathbb{E}[f(X_1^{n-1})]}\\
\Rightarrow H_1^n(F) &\leq \frac{n}{h_2^n(F)}.
\end{align*}

The first line follows because $F$ has decreasing density. Therefore $f(0) \geq f(x)$ for all $x$, and certainly $f(0) \geq \mathbb{E}[f(X)]$ for any non-negative random variable $X$ (including $X_1^{n-1}$). The second line follows by simple algebra. The third line makes two steps. On the LHS, we observe that $f(0) = h_F(0)$, so the left-hand sides are actually identical between the second and third lines. On the right-hand side, we just multiply the numerator and denominator by $n$. The final implication again makes two steps. On the left-hand side, we observe that as $F$ is MHR, $1/h_F(0) \geq 1/h_F(x)$ for all $x \geq 0$. Therefore, $1/h_F(0) \geq \mathbb{E}[1/h_F(X)]$ for any non-negative random variable $X$ (including $X_1^{n}$). On the right-hand side, we have used the equality $n\mathbb{E}[f(X_1^{n-1})] = h_2^n(F)$, which will be proved shortly (and complete the proof). 

The last line above completes the proof (once we establish the equality): we have shown that if $F$ is MHR with decreasing density, then the Limit-Entry Condition is satisfied. The remaining task is to prove Lemma~\ref{lem:h2}, below.

\begin{lemma}\label{lem:h2} $h^2_n(F) = \mathbb{E}[n\cdot f(X_1^{n-1}(F))]$.
\end{lemma}
\begin{proof}

\begin{align*}
h_2^n(F) &= \int_0^\infty n(n-1) f(x) h_F(x) F(x)^{n-2} (1-F(x))dx\\
&=\int_0^\infty n(n-1) f(x) \cdot f(x) F(x)^{n-2}dx\\
&= \mathbb{E}[n\cdot f(X_1^{n-1}(F))].
\end{align*}

The first line is simply the definition of $h_2^n(F)$. The second line just rewrites $h_F(x) (1-F(x)) = f(x)$. The third line observes that $(n-1) f(x) F(x)^{n-2}dx$ is the density of $X_1^{n-1}(F)$. Indeed, there are $n-1$ ways to choose a provider $a$ from $n-1$, $f(x)$ is the density of $v_a$ at $x$, and $F(x)^{n-2}$ is the probability that all $n-2$ other providers have $v_i \leq x$. So we are integrating the density of $X_1^{n-1}(F)$ at $x$, times $f(x)$ from $0$ to $\infty$. This exactly computes the expected value of $f(X_1^{n-1})$. The extra factor of $n$ is carried through.
\end{proof}
\end{proof}
\section{Extension: Non-Zero Costs and Common Base-Value}
\label{sec:costs}
In this section, we establish that our previous results hold verbatim when the provider faces a non-zero cost to provide for the consumer, or in the Common Base-Value model of~\cite{ChawlaMS15}. This section uses standard tricks from auction design to redefine the buyer's value distribution as the seller's profit distribution.Usually, these tricks adjust the distribution of study (e.g. by subtracting a constant). In our setting, due to the mandatory purchase, they actually have no impact on the distribution at all, and our results are independent of the (symmetric) per-consumer cost. We begin by defining non-zero cost.

\begin{definition}[Non-zero cost] Keep all aspects of the model the same, and additionally define a cost $c$ that each provider must pay per consumer covered. That is, if a provider covers a $y$ fraction of the market, and sets price $p$, that provider's payoff is $(p-c)y$.
\end{definition}

To see why our model is independent of the cost $c$, consider the following thought experiment. Instead of having the providers pay cost $c$ per consumer, tell every consumer that they must pay an additional $c$, no matter which provider they choose. Importantly, because purchase is mandatory, this additional cost $c$ doesn't affect the consumer's decisions at all. So consumer decisions are exactly the same as in the zero-cost model. Moreover, because providers are paid an additional $c$ by the consumer to cover the cost, their payoff at a given price vector $\vec{p}$ is exactly the same as if there were zero cost. We again emphasize that the key difference between our setting and typical single-item settings (where cost $c$ with value distribution $F$ is equivalent to zero cost with value distribution $F$ subtracting $c$) is the mandatory purchase. In single-item sale, the ``don't purchase'' option becomes more attractive when the ``purchase'' option has cost increased by $c$. But in our setting, there is no ``don't purchase'' option, and all options uniformly increase in cost by $c$.

Finally, it is easy to see that any strategy profile $\vec{p}$ in the above-described thought experiment (where the consumer pays an additional $c$, independent of choice) is identical to the strategy profile $\vec{p} + \vec{c}$ (i.e. increase all prices by $c$) if the providers pay cost $c$ per consumer covered. This is again because the consumer's choice is invariant under increasing all prices by the same fixed amount. This allows us to conclude the following:

\begin{proposition}\label{prop:costs} All previous results stated when the providers have zero cost hold verbatim when providers have non-zero cost. Specifically:
\begin{itemize}
\item If $D:=F^n$ has a symmetric equilibrium in the Free Market setting with zero cost, then for all $c$, $D:=F^n$ has a symmetric equilibrium in the Free Market setting with cost $c$ (and the equilibrium adds $c$ to all prices).
\item If $D:=F^n$ admits a symmetric equilibrium in the Free Market setting with zero cost, then for any $c$, the expected consumer utility at the unique symmetric equilibrium in the Limited Entry setting exceeds that at the unique symmetric equilibrium in the Free Market setting if and only if $F$ satisfies the Limit-Entry Condition.
\end{itemize}
\end{proposition}

It is easy to see that all other results for the zero-cost setting extend to any cost $c$ by the above two bullets. The same conclusions hold for an extension to Common Base-Value distributions as well.   Note that the Common Base-Value setting may be relevant in our motivating example of health insurance, where a consumer has a base-value for being covered by any insurance at all.

\begin{definition}[Common Base-Value~\cite{ChawlaMS15}] A consumer's valuation vector $\vec{v}$ is drawn from a common base-value distribution $D:= F_0 \times F^n$ if they first draw $\langle w_0,\ldots, w_n\rangle\leftarrow D$ and then set $v_i:= w_0 + w_i$.   
\end{definition}

Similarly to the above reasoning on costs, observe that a consumer with values $\vec{w}$ or with values $\vec{v} = \vec{w} + w_0 \cdot \vec{1}$ will purchase \emph{exactly the same plan} at prices $\vec{p}$ (again, because ``not purchase'' is not an option). Therefore, the common base-value $w_0$ can simply be treated as $0$, because it does not impact any consumer choices, and all previous results immediately extend to this model as well.

\begin{proposition}\label{prop:costs} All previous results stated when the consumers are drawn from a product distribution hold verbatim when consumers are instead drawn from a common base-value distribution. Specifically:
\begin{itemize}
\item If $D:=F^n$ has a symmetric equilibrium in the Free Market setting, then for all base-value distributions $F_0$, $D':=F_0 \times F^n$ has a symmetric equilibrium in the Free Market setting. (It's the same equilibrium.)
\item If $D:=F^n$ admits a symmetric equilibrium in the Free Market setting, then for all base-value distributions $F_0$, the expected consumer utility at the unique symmetric equilibrium in the Limited Entry setting for $D':=F_0 \times F^n$ exceeds that at the unique symmetric equilibrium in the Free Market setting for $D':=F_0 \times F^n$ if and only if $F$ satisfies the Limit-Entry Condition.
\end{itemize}
\end{proposition}

\section{Extension: Asymmetric Distributions} \label{sec:asym}
In this section, we extend our results on equilibria in the Free Market setting to asymmetric distributions. Here, we only prove existence results for certain kinds of equilibria rather than closed form solutions, hence it would not be possible to extend the Limit-Entry Condition without a significantly different approach. To be clear: we are in exactly the same model as the rest of the paper, but no longer assuming that $D$ is symmetric. We will still aim for sufficient conditions for a symmetric equilibrium to exist, but our extensions hold for restricted cases. Many of the lemmas in this section hold for general $n$, and we will state them as such.

The structure of this section will closely parallel that of Section~\ref{sec:design}, and the proofs follow similar intuition with additional technical work. We begin by updating the definition of the star operation.

\begin{definition}[Star Operation] Let $F^*_{i,\vec{p}_{-i}}(q)$ be such that when each provider $j \neq i$ sets price $p_j$, and provider $i$ sets price $q$, the probability that the consumer purchases from provider $i$ is $1-F^*_{i,\vec{p}_{-i}}(q)$. That is, the payoff to provider $i$ in this circumstance is $q(1-F^*_{i,\vec{p}_{-i}}(q))$. 
\end{definition}

As in Section~\ref{sec:design}, we now compute $F^*_{i,\vec{p}_{-i}}(q)$ (and the related quantities).  A proof is provided in Appendix~\ref{proof:asymApp}.

\begin{proposition}\label{prop:asymstars}Let $g_{i,\vec{p}_{-i}}(q,x):= \sum_{j \neq i} f_j(x-q+p_j) \prod_{k \notin \{i,j\}} F_k(x-q+p_k)$. Let also $M:= \max\{0, q-\min_{j \neq i}\{p_j\}\}$. Then:
\begin{itemize}
\item $F^*_{i,\vec{p}_{-i}}(q) = \int_M^\infty F_i(x) g_{i,\vec{p}_{-i}}(q,x)dx$.
\item $1-F^*_{i,\vec{p}_{-i}}(q) = \int_M^\infty (1-F_i(x)) g_{i,\vec{p}_{-i}}(q,x)dx+\prod_{j \neq i} F_j(M + \min_{j \neq i}\{p_j\} - q)$.
\item $f^*_{i,\vec{p}_{-i}}(q) = \int_M^\infty f_i(x) g_{i,\vec{p}_{-i}}(q,x)dx$.
\item $(f^*_{i,\vec{p}_{-i}})'(q) = \int_M^\infty (f'_i(x)) g_{i,\vec{p}_{-i}}(q,x) dx+f_i(0)g_{i,\vec{p}_{-i}}(q,M)$.
\end{itemize}
\end{proposition}

Like in Section~\ref{sec:design}, we establish that if $F_i$ is \mhrplus\ and decreasing density, then $F_{i,\vec{p}_{-i}}^*$ is MHR, for all $\vec{p}_{-i}$. 

\begin{proposition}\label{prop:mainstars2} Let $F_i$ be \mhrplus\ and decreasing density. Then for all $\vec{p}_{-i}$, $F^*_{i,\vec{p}_{-i}}$ is MHR.
\end{proposition}

The proof parallels that of Proposition~\ref{prop:mainstars}, and can also be found in Appendix~\ref{proof:asymApp}.  Finally, we use this proposition to establish that an equilibrium exists for all two-provider instances with \mhrplus\ marginals. 

\begin{theorem}
Let $F_1,F_2$ be \mhrplus. Then there exists a pure equilibrium for $D:= F_1 \times F_2$.
\end{theorem}
\begin{proof}
The proof will proceed by applying Brouwer's fixed point theorem. There are two high-level steps. First, we must establish that the best-response function for player $i$ (responding to a price set by the other player) is well-defined and continuous (Lemma~\ref{lem:brcont}). Second, we must establish that it maps a compact domain to itself. The first step is fairly straight-forward, while the second step requires some creativity (and indeed the second step is messy to extend beyond two providers). We establish the first half of the proof for general $n$. We begin first by establishing that the best response function $q_i(\cdot)$ is continuous for all $i$ when $F^*_{i,\vec{p}_{-i}}$ is MHR.

\begin{lemma}\label{lem:brcont} Let $F^*_{i,\vec{p}_{-i}}$ be MHR for all $\vec{p}_{-i}$. Then the best response function $q_i(\cdot)$ which takes as input $\vec{p}_{-i}$ and outputs the best response price for provider $i$ is continuous.
\end{lemma}
\begin{proof}
Because $F^*_{i,\vec{p}_{-i}}$ is MHR for all $\vec{p}_{-i}$, we can actually write a closed form for the best response function $q_i(\vec{p}_{-i})$. Indeed, because there is a unique $q$ satisfying the first order conditions of $\varphi_{F^*_{i,\vec{p}_{-i}}}(q) = 0$, this is the best response. We can therefore write:
\begin{align*}
q_i(\vec{p}_{-i}) &= (\varphi_{F^*_{i,\vec{p}_{-i}}})^{-1}(0),
\end{align*}
recalling from Proposition~\ref{prop:asymstars} that:
$$\varphi_{F^*_{i,\vec{p}_{-i}}}(q):=q-\frac{1- \int_M^\infty F_i(x) g_{i,\vec{p}_{-i}}(q,x)dx}{\int_M^\infty f_i(x) g_{i,\vec{p}_{-i}}(q,x)dx},$$
where $M= \max\{0,q-\min_{j \neq i}\{p_j\}\}$ and $g_{i,\vec{p}_{-i}}(q,x)= \sum_{j \neq i} f_j(x-q+p_j)\prod_{k \notin \{i,j\}} F_k(x-q+p_k)$. As $F^*_{i,\vec{p}_{-i}}$ is MHR, observe that $\varphi_{F^*_{i,\vec{p}_{-i}}}(\cdot)$ is monotone strictly increasing, and also continuous (in fact, it is also differentiable). Therefore, the inverse is well-defined, and also continuous (in fact, differentiable). 

Observe that our function $q_i(\cdot)$ is exactly the function which takes as input $\vec{p}_{-i}$ and outputs the $q$ such that $\varphi_{F^*_{i,\vec{p}_{-i}}}(q) = 0$. As we have noted that $\varphi_{F^*_{i,\vec{p}_{-i}}}(\cdot)$ is continuous
, $q_i(\cdot)$ is continuous as well, by the implicit function theorem.
\end{proof}

Now, we want to consider the function $\vec{q}(\cdot)$ which takes as input a price vector $\vec{p}$ and outputs $\langle q_i(\vec{p}_{-i})\rangle_{i \in [n]}$. This function is continuous, since each coordinate is continuous. We just want to show that it maps a compact region to itself, and then Brouwer's fixed point theorem will establish that an equilibrium exists. 

So our plan is to find a value $T$ such that for all $i$, $p_i \leq T$, $q_i(\vec{p})\leq T$ as well (from here, we will abuse notation and let $q_i(\vec{p})$ denote the $i^{th}$ coordinate of our vector-valued function $\vec{q}(\vec{p})$, which is also equal to the afore-defined $q_i(\vec{p}_{-i})$).  We begin with the following observation.


\begin{observation} \label{obs:mostx} Let $F^*_{i,\vec{p}_{-i}}$ be MHR, and let $\frac{1-F^*_{i,\vec{p}_{-i}}(x)}{f^*_{i,\vec{p}_{-i}}(x)} \leq x$. Then $q_i(\vec{p}) \leq x$.
\end{observation}
\begin{proof}
Clearly, $\varphi_{F^*_{i,\vec{p}_{-i}}}(x) = x - \frac{1-F^*_{i,\vec{p}_{-i}}(x)}{f^*_{i,\vec{p}_{-i}}(x)} \geq 0$. Because $\varphi_{F^*_{i,\vec{p}_{-i}}}(\cdot)$ is monotone increasing, it means that $q_i(\vec{p}) = (\varphi_{F^*_{i,\vec{p}_{-i}}})^{-1}(0) \leq x$.
\end{proof}

\begin{corollary}\label{cor:almost}Let $n=2$. Then $\partial q_i(\vec{p})/\partial p_j\geq 0$ when $i \neq j$.
\end{corollary}
\begin{proof}
Crucially, observe that $g_{i,p+\varepsilon}(q+\varepsilon,x) = g_{i,p}(q,x)$ for all $x,p,q,\varepsilon$. Therefore, $\frac{1-F^*_{i,p}(q)}{f^*_{i,p}(q)} =\frac{1-F^*_{i,p+\varepsilon}(q+\varepsilon)}{f^*_{i,p+\varepsilon}(q+\varepsilon)}$. This should be intuitive: if all prices increase by $\varepsilon$, then all of the probabilities of sale stay the same. This immediately implies that:
$$q_i(\vec{p})  =\frac{1-F^*_{i,p}(q_i(\vec{p}))}{f^*_{i,p}(q_i(\vec{p}))}
=\frac{1-F^*_{i,p+\varepsilon}(q_i(\vec{p})+\varepsilon)}{f^*_{i,p+\varepsilon}(q_i(\vec{p})+\varepsilon)}
\leq \frac{1-F^*_{i,p+\varepsilon}(q_i(\vec{p}))}{f^*_{i,p+\varepsilon}(q_i(\vec{p}))} $$
The last inequality follows as $1/h_{F_{i,p+\varepsilon}^*}(q_i(\vec{p}) + \varepsilon) \leq 1/h_{F_{i,p+\varepsilon}^*}(q_i(\vec{p}))$ since $F^*_{i,p}$ is MHR for any $p$. Let's see now what the last line implies. It means certainly that $\varphi_{F^*_{i,p+\varepsilon}}(q_i(\vec{p})) \leq 0$, and therefore $(\varphi_{F^*_{i,p+\varepsilon}})^{-1}(0) \geq q_i(\vec{p})$. We therefore conclude that $q_i(\vec{p})$ is weakly increasing with respect to the non-$i$ coordinate, proving the corollary.
\end{proof}

Corollary~\ref{cor:almost} establishes $q_i(\vec{p})$ is non-decreasing in the other provider's price, and further, that when $\vec{p} \in [0,T]^2$, then $q_i(\vec{p})$ is maximized when $p_j = T$ for $j \neq i$, and that $q_i(\vec{p}) \leq \frac{1-F^*_{i,T}(T)}{f^*_{i,T}(T)}$.  Then to conclude that we map a compact region to itself, we just need to show that $\frac{1-F^*_{i,T}(T)}{f^*_{i,T}(T)} \leq T$.


\begin{lemma}\label{lem:there!}For any $n$, let $F^*_{i,\vec{p}_{-i}}$ be MHR for all $i,\vec{p}_{-i}$. Let also $T \geq\frac{1-F^*_{i,\vec{0}}(0)}{f^*_{i,\vec{0}}(0)}$. Then $T \geq \frac{1-F^*_{i,T\cdot \vec{1}}(T)}{f^*_{i,T\cdot \vec{1}}(T)}$.
\end{lemma}
\begin{proof}
Observe that in fact $\frac{1-F^*_{i,T\cdot \vec{1}}(T)}{f^*_{i,T\cdot \vec{1}}(T)}$ is independent of $T$. To see this, recall that $F^*_{i,T\cdot \vec{1}}(T)$ concerns only the probability that the consumer purchases provider $i$ when all prices are $T$ (and $f^*_{i,T\cdot \vec{1}}(T)$ is its derivative with respect to $p_i$). Observe, however, that because the buyer must purchase an option even when their utility is negative for everything, that this probability is completely independent of $T$. Therefore, this is actually equal to $\frac{1-F^*_{i,\vec{0}}(0)}{f^*_{i,\vec{0}}(0)}$, which is just some finite number. Whenever $T \geq \frac{1-F^*_{i,\vec{0}}(0)}{f^*_{i,\vec{0}}(0)}$, the condition is satisfied.
\end{proof}

Finally, we may now simply set $T:= \max_{i \in \{1,2\}}\{\frac{1-F^*_{i,\vec{0}}(0)}{f^*_{i,\vec{0}}(0)}\}$ to guarantee that it is at least as large as both. By the work above, $\vec{q}(\cdot)$ maps $[0,T]^2$ to $[0,T]^2$, and is continuous. Therefore it has a fixed point by Brouwer's fixed point theorem, which is exactly a pure equilibrium.
\end{proof}

\bibliographystyle{plainnat}
\bibliography{local}

\newpage
\appendix
\section{Omitted Proofs from Section~\ref{sec:design}}\label{app:proofs}
\begin{numberedproposition}{\ref{prop:stars}} Let $g_p(q,x):= (n-1)f(x-q+p)(F(x-q+p))^{n-2}$. Let also $M:= \max\{0, q-p\}$.
\begin{itemize}
\item $F^*_p(q) = \int_M^\infty F(x) g_p(q,x)dx$.
\item $1-F^*_p(q) = \int_M^\infty (1-F(x)) g_p(q,x)dx+F(M+p-q)^{n-1}$.
\item $f^*_p(q) = \int_M^\infty f(x) g_p(q,x)dx$.
\item $(f^*_p)'(q) = \int_M^\infty (f'(x)) g_p(q,x) dx+f(0)g_p(q,M)$.
\end{itemize}
\end{numberedproposition}
\begin{proof}

For the first bullet, $F^*_p(q)$ denotes the probability that the consumer does not have highest utility for the lone provider ($n$) setting price $q$. Observe that $g_p(q,x)$ is the derivative of $F(x-q+p)^{n-1}$ with respect to $x$. If $X:= \max_{i \leq n-1}\{v_i\} +q-p$ denotes the random variable that draws $n-1$ times from $F$ and takes the maximum, then adds $q-p$, then $F(x-q+p)^{n-1}$ is the CDF of $X$, so $g_p(q,x)$ is the PDF of $X$. Observe also that the consumer will not purchase from the lone provider setting price $q$ iff $v_n-q \leq \max_{i \leq n-1}\{v_i\} - p \Leftrightarrow v_n \leq \max_{i \leq n-1} \{v_i\}+q-p \Leftrightarrow v_n \leq X$. 

One way to compute the probability that $v_n \leq X$ is to integrate over all possible values of $X$ ($x$), the density of $X$ at $x$ times the probability that $v_n$ does not exceed $x$. When $x < 0$, the probability that $v_n$ does not exceed $x$ is $0$. When $x < q-p$, the density of $X$ at $x$ is $0$. So we may restrict the integral to the range $[M,\infty)$. Finally, $F(x)$ is exactly the probability that $v_n$ does not exceed $x$, and $g_p(q,x)$ is exactly the density of $X$ at $x$, yielding the first bullet.

For the second bullet, observe that $\int_M^\infty g_p(q,x)dx$ is integrating the density of some random variable from $M$ to $\infty$. This random variable is supported on $q-p$ to $\infty$. So if $M = q-p$, the integral is $1$, and we have that:

$$1-F^*_p(q) = 1 - \int_M^\infty F(x) g_p(q,x)dx = \int_M^\infty (1-F(x)) g_p(q,x)dx$$

If $M = 0$, then the integral isn't necessarily $1$, but is instead $1-F(p-q)^{n-1}$, so an additional $F(p-q)^{n-1}$ needs to be added. Observe that $F(M+p-q)^{n-1} = 0$ when $M \neq 0$, and is $F(p-q)^{n-1}$ otherwise, as desired.

For the third bullet, we apply Leibniz' integral rule, and take the derivative of $F^*_p(q)$ with respect to $q$. The derivative of $\infty$ with respect to $q$ is $0$, and the derivative of $M$ with respect to $q$ is $\mathbb{I}(M > 0)$. So we get that:

$$f^*_p(q) = \int_M^\infty F(x) \frac{\partial g_p(q,x)}{\partial q} dx - \mathbb{I}(M > 0) F(M)g_p(q,M).$$

Doing integration by parts, with $u:= F(x)$ and $dv:= \frac{\partial g_p(q,x)}{\partial q}dx$, and observing, crucially, that $\frac{\partial g_p(q,x)}{\partial q} = - \frac{\partial g_p(q,x)}{\partial x}$ (so $du = f(x) dx$ and $v = -g_p(q,x)$), we get:

$$f^*_p(q) = F(x)g_p(q,x) |^\infty_M +\int_M^\infty f(x) g_p(q,x) dx - \mathbb{I}(M > 0) F(M)g_p(q,M).$$
$$f^*_p(q) = \int_M^\infty f(x) g_p(q,x) dx + F(M)g_p(q,M)- \mathbb{I}(M > 0) F(M)g_p(q,M).$$

Finally, observe that if $M> 0$, the terms at the right cancel. if $M = 0$, both terms are $0$ (and still cancel). 

For the fourth bullet, we again apply Leibniz' integral rule, and take the derivative with respect to $q$. We get:

$$(f^*)'_p(q) = \int_M^\infty f(x) \frac{\partial g_p(q,x)}{\partial q} dx -\mathbb{I}(M > 0)f(M) g_p(q,M).$$

Again doing integration by parts, with $u:= f(x)$ and $dv:= \frac{\partial g_p(q,x)}{\partial q}dx$ (so $du = f'(x)dx$, $v= -g_p(q,x)$), we get:

$$(f^*)'_p(q) = f(x)g_p(q,x) |^\infty_M +\int_M^\infty f'(x) g_p(q,x) dx - \mathbb{I}(M > 0) f(M)g_p(q,M).$$
$$(f^*)'_p(q) = \int_M^\infty f(x) g_p(q,x) dx + f(M)g_p(q,M)- \mathbb{I}(M > 0) f(M)g_p(q,M).$$

Finally, observe that if $M > 0$, then the right two terms cancel. If $M=0$, then we're left with $f(0)g_p(q,M)$. Observe that when $M \neq 0$, $g_p(q,M) = 0$, so the added term in the proposition statement is correct.
\end{proof}

\section{Omitted Proofs from Section~\ref{sec:asym}}\label{proof:asymApp}

\begin{numberedproposition}{\ref{prop:asymstars}}Let $g_{i,\vec{p}_{-i}}(q,x):= \sum_{j \neq i} f_j(x-q+p_j) \prod_{k \notin \{i,j\}} F_k(x-q+p_k)$. Let also $M:= \max\{0, q-\min_{j \neq i}\{p_j\}\}$. Then:
\begin{itemize}
\item $F^*_{i,\vec{p}_{-i}}(q) = \int_M^\infty F_i(x) g_{i,\vec{p}_{-i}}(q,x)dx$.
\item $1-F^*_{i,\vec{p}_{-i}}(q) = \int_M^\infty (1-F_i(x)) g_{i,\vec{p}_{-i}}(q,x)dx+\prod_{j \neq i} F_j(M + \min_{j \neq i}\{p_j\} - q)$.
\item $f^*_{i,\vec{p}_{-i}}(q) = \int_M^\infty f_i(x) g_{i,\vec{p}_{-i}}(q,x)dx$.
\item $(f^*_{i,\vec{p}_{-i}})'(q) = \int_M^\infty (f'_i(x)) g_{i,\vec{p}_{-i}}(q,x) dx+f_i(0)g_{i,\vec{p}_{-i}}(q,M)$.
\end{itemize}
\end{numberedproposition}

\begin{proof}
For the first bullet, $F^*_{i,\vec{p}_{-i}}(q)$ denotes the probability that the consumer does not have highest utility for the provider $i$ setting price $q$. Observe that $g_{i,\vec{p}_{-i}}(q,x)$ is the derivative of $\prod_{k \neq i} F_k(x-q+p_k)$ with respect to $x$. If $X:= \max_{j \neq i}\{v_j +q - p_j\} +q-p$ denotes the random variable that draws from each of $F_j$ ($j \neq i$), then adds $q - p_j$ and takes the maximum, then $\prod_{k \neq i} F_k(x - q + p_k)$ is the CDF of $X$, so $g_{i,\vec{p}_{-i}}(q,x)$ is the PDF of $X$. Observe also that the consumer will not purchase from provider $i$ iff $v_i-q \leq \max_{j \neq i}\{v_j - p_i\}  \Leftrightarrow v_i \leq \max_{j \neq i} \{v_j+ q - p_j\} \Leftrightarrow v_i \leq X$. 

One way to compute the probability that $v_i \leq X$ is to integrate over all possible values of $X$ ($x$), the density of $X$ at $x$ times the probability that $v_i$ does not exceed $x$. When $x < 0$, the probability that $v_i$ does not exceed $x$ is $0$. When $x < \min_{j \neq i} \{p_j - q\}$, the density of $X$ at $x$ is $0$. So we may restrict the integral to the range $[M,\infty)$. Finally, $F_i(x)$ is exactly the probability that $v_i$ does not exceed $x$, and $g_{p,\vec{p}_{-i}}(q,x)$ is exactly the density of $X$ at $x$, yielding the first bullet.

For the second bullet, observe that $\int_M^\infty g_{i,\vec{p}_{-i}}(q,x)dx$ is integrating the density of some random variable from $M$ to $\infty$. This random variable is supported on $q-\min_{j \neq i}\{p_j\}$ to $\infty$. So if $M = q-\min_{j \neq i}\{p_j\}$, the integral is $1$, and we have that:

$$1-F^*_{i,\vec{p}_{-i}}(q) = 1 - \int_M^\infty F_i(x) g_{i,\vec{p}_{-i}}(q,x)dx = \int_M^\infty (1-F_i(x)) g_{i, \vec{p}_{-i}}(q,x)dx$$

If $M = 0$, then the integral isn't necessarily $1$, but is instead $1-\prod_{j \neq i} F_j(\min_{j \neq i}\{p_j\} - q)$, so an additional $\prod_{j \neq i} F_j(\min_{j \neq i}\{p_j\} - q)$ term needs to be added. Observe that when $M = 0$, this is exactly the term added in the statement. When $M = q - \min_{j \neq i} \{p_j\}$, the added term is $0$ (and should be $0$, by above).

For the third bullet, we apply Leibniz' integral rule, and take the derivative of $F^*_{i,\vec{p}_{-i}}(q)$ with respect to $q$. The derivative of $\infty$ with respect to $q$ is $0$, and the derivative of $M$ with respect to $q$ is $\mathbb{I}(M > 0)$. So we get that:

$$f^*_{i,\vec{p}_{-i}}(q) = \int_M^\infty F_i(x) \frac{\partial g_{i,\vec{p}_{-i}}(q,x)}{\partial q} dx - \mathbb{I}(M > 0) F_i(M)g_{i,\vec{p}_{-i}}(q,M).$$

Doing integration by parts, with $u:= F_i(x)$ and $dv:= \frac{\partial g_{i,\vec{p}_{-i}}(q,x)}{\partial q}dx$, and observing, crucially, that $\frac{\partial g_{i,\vec{p}_{-i}}(q,x)}{\partial q} = - \frac{\partial g_{i,\vec{p}_{-i}}(q,x)}{\partial x}$ (so $du = f_i(x) dx$ and $v = -g_{i,\vec{p}_{-i}}(q,x)$), we get:

$$f^*_{i,\vec{p}_{-i}}(q) = F_i(x)g_{i,\vec{p}_{-i}}(q,x) |^\infty_M +\int_M^\infty f_i(x) g_{i,\vec{p}_{-i}}(q,x) dx - \mathbb{I}(M > 0) F_i(M)g_{i,\vec{p}_{-i}}(q,M).$$
$$f^*_i(q) = \int_M^\infty f_i(x) g_{i,\vec{p}_{-i}}(q,x) dx + F_i(M)g_{i,\vec{p}_{-i}}(q,M)- \mathbb{I}(M > 0) F_i(M)g_{i,\vec{p}_{-i}}(q,M).$$

Finally, observe that if $M> 0$, the terms at the right cancel. if $M = 0$, both terms are $0$ (and still cancel). 

For the fourth bullet, we again apply Leibniz' integral rule, and take the derivative with respect to $q$. We get:

$$(f^*)'_{i,\vec{p}_{-i}}(q) = \int_M^\infty f_i(x) \frac{\partial g_{i,\vec{p}_{-i}}(q,x)}{\partial q} dx -\mathbb{I}(M > 0)f_i(M) g_{i,\vec{p}_{-i}}(q,M).$$

Again doing integration by parts, with $u:= f_i(x)$ and $dv:= \frac{\partial g_{i,\vec{p}_{-i}}(q,x)}{\partial q}dx$ (so $du = f'_i(x)dx$, $v= -g_{i,\vec{p}_{-i}}(q,x)$), we get:

$$(f^*)'_i(q) = f_i(x)g_{i,\vec{p}_{-i}}(q,x) |^\infty_M +\int_M^\infty f'_i(x) g_{i,\vec{p}_{-i}}(q,x) dx - \mathbb{I}(M > 0) f_i(M)g_{i,\vec{p}_{-i}}(q,M).$$
$$(f^*)'_{i,\vec{p}_{-i}}(q) = \int_M^\infty f_i(x) g_{i,\vec{p}_{-i}}(q,x) dx + f_i(M)g_{i,\vec{p}_{-i}}(q,M)- \mathbb{I}(M > 0) f_i(M)g_{i,\vec{p}_{-i}}(q,M).$$

Finally, observe that if $M > 0$, then the right two terms cancel. If $M=0$, then we're left with $f_i(0)g_{i,\vec{p}_{-i}}(q,M)$. Observe that when $M \neq 0$, $g_{i,\vec{p}_{-i}}(q,M) = 0$, so the added term in the proposition statement is correct.
\end{proof}

We continue with another missing proof that parallels Section~\ref{sec:asym}.

\begin{numberedproposition}{\ref{prop:mainstars2}} Let $F_i$ be \mhrplus\ and decreasing density. Then for all $\vec{p}_{-i}$, $F^*_{i,\vec{p}_{-i}}$ is MHR.
\end{numberedproposition}

\begin{proof}
First, observe that if $F_i$ has decreasing density, then $f_i(0) > 0$ (so we may divide by $f_i(0)$). Next, recall by Observation~\ref{obs:defs} that we have $f_i(x) f_i(0) \geq -f'_i(x)$, and $h_{F_i}(x) \geq f_i(0) \Rightarrow f_i(x) \geq f_i(0) (1-F_i(x))$. Therefore, we get:

\begin{align*}
1-F^*_{i,\vec{p}_{-i}}(q) &= \int_M^\infty (1-F_i(x)) g_{i,\vec{p}_{-i}}(q,x) dx + \prod_{j \neq i} F_j(M + \min_{j \neq i}\{p_j\} - q)\\
&\leq \int_M^\infty \frac{f_i(x)}{f_i(0)} g_{i,\vec{p}_{-i}}(q,x) dx + \prod_{j \neq i} F_j(M + \min_{j \neq i}\{p_j\} - q)\\
&= f^*_{i,\vec{p}_{-i}}(q)/f_i(0) + \prod_{j \neq i} F_j(M + \min_{j \neq i}\{p_j\} - q).
\end{align*}

Similarly, we can write:

\begin{align*}
-(f^*_{i,\vec{p}_{-i}})'(q) &= -\int_M^\infty f'_i(x) g_{i,\vec{p}_{-i}}(q,x) dx - f_i(0)g_{i,\vec{p}_{-i}}(q,M)\\
&\leq \int_M^\infty f_i(x)f_i(0)g_{i,\vec{p}_{-i}}(q,x)dx - f_i(0) g_{i,\vec{p}_{-i}}(q,M)\\
&\leq f_i(0) f^*_{i,\vec{p}_{-i}}(q) - f_i(0) g_{i,\vec{p}_{-i}}(q,M).
\end{align*}

Therefore, we get:

\begin{align*}
(1-F^*_{i,\vec{p}_{-i}}(q)) \cdot (-f^*_{i,\vec{p}_{-i}})'(q) &\leq (f^*_{i,\vec{p}_{-i}}(q))^2 + f_i(0)g_{i,\vec{p}_{-i}}(q,M)\cdot (1-F^*_{i,\vec{p}_{-i}}(q)) \\&\quad- \prod_{j \neq i} F_j(M+\min_{j \neq i}\{p_j\} - q) \int_M^\infty f'_i(x) g_{i,\vec{p}_{-i}}(q,x)dx.
\end{align*}

So if we can show that

$$ f_i(0)g_{i,\vec{p}_{-i}}(q,M)\cdot (1-F^*_{i,\vec{p}_{-i}}(q))- \prod_{j \neq i} F_j(M+\min_{j \neq i}\{p_j\} - q) \int_M^\infty f'_i(x) g_{i,\vec{p}_{-i}}(q,x)dx\leq 0,$$

then we will have established that $(1-F^*_{i,\vec{p}_{-i}}(q)) \cdot (-f^*_{i,\vec{p}_{-i}})'(q) \leq (f^*_{i,\vec{p}_{-i}}(q))^2$, and therefore $F^*_{i,\vec{p}_{-i}}$ is MHR. Observe that this is clearly true when $M=q-\min_{j \neq i}\{p_j\}$, as the entire term above is $0$ (because $g_{i,\vec{p}_{-i}}(q,q-\min_{j \neq i} \{p_j\}) = 0$ and $F_j(0) = 0$ for all $j$). So the remaining case is when $M =0$. Here, we derive the following (justification for each equation follows). 

\begin{align*}
-\int_0^\infty f'_i(x) g_{i,\vec{p}_{-i}}(q,x)dx &\leq \int_0^\infty f_i(0)f_i(x) g_{i,\vec{p}_{-i}}(q,x) dx\\
\Rightarrow \frac{-\int_0^\infty f'_i(x) g_{i,\vec{p}_{-i}}(q,x)dx}{f_i(0)g_{i,\vec{p}_{-i}}(q,0)}&\leq \int_0^\infty f_i(x)g_{i,\vec{p}_{-i}}(q,x)/g_{i,\vec{p}_{-i}}(q,0)dx\\
&= \int_0^\infty f_i(x) \frac{ \sum_{j \neq i} f_j(x-q+p_j) \prod_{k \notin \{i,j\}} F_k(x-q+p_k)}{ \sum_{j \neq i} f_j(p_j-q) \prod_{k \notin \{i,j\}} F_k(p_k-q)}dx
\end{align*}
\begin{align*}
\Rightarrow  &\frac{-\prod_{k \neq i}F_k(p_k-q)\int_0^\infty f'_i(x) g_{i,\vec{p}_{-i}}(q,x)dx}{f_i(0)g_{i,\vec{p}_{-i}}(q,0)}\\
 & \leq \int_0^\infty f_i(x) \frac{ \sum_{j \neq i} f_j(x-q+p_j)F_j(p_j-q) \prod_{k \notin \{i,j\}} F_k(x-q+p_k)}{ \sum_{j \neq i} f_j(p_j-q)}dx\\
&\leq \int_0^\infty f_i(x) \frac{ \sum_{j \neq i} f_j(x-q+p_j)\prod_{k \neq i} F_k(x-q+p_k)}{ \sum_{j \neq i} f_j(p_j-q)}dx\\
&\leq \int_0^\infty f_i(x) \prod_{k \neq i} F_k(x-q+p_k)dx\\
&= 1-F_p^*(q).
\end{align*}

The first inequality follows by definition of $F$ being \mhrplus. The second follows by dividing both sides by $f_i(0)g_{i,\vec{p}_{-i}}(q,0)$, which is positive. The third line follows by evaluating the definition of $g_{i,\vec{p}_{-i}}(q,x)$.

The fourth line then follows by multiplying both sides by $\prod_{k\neq i}F_k(p_k-q)$, which is positive. The fifth line follows as for all $j$, $F_j(p_j-q) \leq F_j(x-q+p_j)$. The penultimate line follows because for all $j$, $F_j$ has decreasing density (and therefore $f_j(x-q+p_j) \leq f_j(p_j-q)$ for all $j$). 

The final line follows from the following reasoning. Recall that $1-F_{i,\vec{p}_{-i}}^*(q)$ denotes the probability that the buyer will choose to purchase item $i$ when item $i$ has price $q$ and all other items have price $\vec{p}_{-i}$. The probability that this occurs conditioned on $v_i = x$ is exactly $\prod_{k \neq i}F_k(x+p_k-q)$, and the previous term simply integrates this times $f_i(x)$ over all $x$.

Finally, observe that this inequality is exactly what we want, as (the first line below is exactly what we just proved above, and the final line is our remaining task).
\begin{align*}
\frac{-F^{n-1}(p-q)\int_0^\infty f'_i(x) g_p(q,x)dx}{f_i(0)g_p(q,0)} &\leq 1-F^*_p(q) \\
\Rightarrow -F^{n-1}(p-q)\int_0^\infty f'_i(x) g_p(q,x)dx &\leq f_i(0)g_p(q,M)\cdot (1-F^*_p(q))\\
\Rightarrow -f_i(0)g_p(q,M)\cdot (1-F^*_p(q)) -F(M+p-q)^{n-1}\int_M^\infty f'_i(x) g_p(q,x)dx &\leq 0.
\end{align*}

\end{proof}
\section{Calculations for Examples}
\label{app:examples}


Recall that $F^*_{p}(q)$ denotes the probability that a consumer drawn from $D=F^n$ does not have the highest utility for the lone provider $(n)$ setting price $q$, when the other providers $1,\dots,n-1$ are setting price p. Then, the expected payoff of the lone provider (n) in this circumstance is $q \cdot (1-F^*_{p}(q))$.\\

Consider the following class of examples, which we'll refer to as $F_{\varepsilon,k}$. Note that the hazard rate is the most interesting part of the distribution, and the CDF ensures such a distribution is well-defined.
\normalsize

$h_{F_{\varepsilon,k}}(x) = \varepsilon$ \ \  x.p. $\frac{1}{k}$ \ \ \ \ \ \ \ \ \ \ \ \ \ \ \ \ \  
$F_{\varepsilon,k}(x) = 1-e^{-\varepsilon x}$ \ \ \ \ \  \ \ \ \ \ \  
$f_{\varepsilon,k}(x) = \varepsilon e^{-\varepsilon x}$ \ \ \ \  for $x \leq \frac{ln(k)}{\varepsilon}$\\

$h_{F_{\varepsilon,k}}(x) = 1$ \ \  x.p. $1-\frac{1}{k}$
\ \ \ \ \ \ \ \ \ 
$F_{\varepsilon,k}(x) = 1-(\frac{1}{k})^{1 - \frac{1}{\varepsilon}}e^{-x}$ \ \ 
\ \ \ \ 
$f_{\varepsilon,k}(x) = (\frac{1}{k})^{1 - \frac{1}{\varepsilon}}e^{-x}$ \ \ for $x > \frac{ln(k)}{\varepsilon}$\\


\begin{observation}
For all $\varepsilon \in (0,1], k \geq 1$, $F_{\varepsilon,k}$ is MHR.
\end{observation}


\begin{lem}
There exists $\varepsilon \in (0,1), \ k > 1$, \ such that $F^*_{p,(\varepsilon,k)}$ is anti-MHR.
\end{lem}

\begin{proof}

We start by defining an even more general class of MHR distributions. In particular, let $c_1 < c_2$ be any two positive constants, and let $F$ be any distribution that satisfies the following properties:

\ \ \ \ \ \ \ \ $1-F(x) = \frac{f(x)}{c_1}$ and $-f'(x) = c_1f(x)$ for $x \in [\max\{0, q-p\}, x_1]$  

and \ \ $1-F(x) = \frac{f(x)}{c_2}$ and $-f'(x) = c_2f(x) $ for $x \in [x_1, \infty)$\\

Clearly, \textbf{$F$ is MHR} because $\forall x$, $-f'(x)[1-F(x)] \leq f(x)^2$. We now compute $-f^{'*}_p(q)[1-F^*_p(q)] $ to see whether $F^*_p(q)$ is MHR or not.\\

$1-F^*_p(q) = \int_{M}^{x_1} g_p(q,x)[1-F(x)]dx + F^{n-1}(M+p-q)$ 

$= \int_{M}^{x_1} g_p(q,x)\frac{f(x)}{c_1}dx + \int_{x_1}^{\infty} g_p(q,x)\frac{f(x)}{c_2}dx + F^{n-1}(M+p-q) $\\

$-f^{'*}_p(q) = \int_{M}^{\infty} g_p(q,x)[-f'(x)]dx - f(0)g_p(q,M)$

$= \int_{M}^{x_1} g_p(q,x)c_1f(x)dx + \int_{x_1}^{\infty} g_p(q,x)c_2f(x)dx - f(0)g_p(q,M)$\\

$-f^{'*}_p(q)[1-F^*_p(q)] $ 

$= \int_{M}^{x_1} g_p(q,x)\frac{f(x)}{c_1}dx \cdot \int_{M}^{x_1} g_p(q,x)c_1f(x)dx$

$+ \int_{M}^{x_1} g_p(q,x)\frac{f(x)}{c_1}dx \cdot \int_{x_1}^{\infty} g_p(q,x)c_2f(x)dx$

$+ \int_{x_1}^{\infty} g_p(q,x)\frac{f(x)}{c_2}dx \cdot \int_{M}^{x_1} g_p(q,x)c_1f(x)dx$

$+ \int_{x_1}^{\infty} g_p(q,x)\frac{f(x)}{c_2}dx \cdot \int_{x_1}^{\infty} g_p(q,x)c_2f(x)dx$

$+ [F^{n-1}(M+p-q)][-f^{'*}_p(q)] - [1-F^*_p(q)][f(0)g_p(q,M)] - [F^{n-1}(M+p-q)][f(0)g_p(q,M)] $\\

Note this last term can be reformulated as:
$-a = [F^{n-1}(M+p-q)][-f^{'*}_p(q)] - [1-F^*_p(q)][f(0)g_p(q,M)] - [F^{n-1}(M+p-q)][f(0)g_p(q,M)]$ 
$\geq -2[f(0)g_p(q,M)] = -2c_1 g_p(q,M)$ and is thus lower bounded by a constant. We can now rewrite the product as follows:\\

$-f^{'*}_p(q)[1-F^*_p(q)] $ 

$= \int_{M}^{x_1} g_p(q,x)f(x)dx \cdot \int_{M}^{x_1} g_p(q,x)f(x)dx$

$+ \frac{c_2}{c_1} \int_{M}^{x_1} g_p(q,x)f(x)dx \cdot \int_{x_1}^{\infty} g_p(q,x)f(x)dx$

$+ \frac{c_1}{c_2} \int_{x_1}^{\infty} g_p(q,x)f(x)dx \cdot \int_{M}^{x_1} g_p(q,x)f(x)dx$

$+ \int_{x_1}^{\infty} g_p(q,x)f(x)dx \cdot \int_{x_1}^{\infty} g_p(q,x)f(x)dx$

$-a$ \ \ \ \  (assume worst case $a \geq 0$)\\

$-f^{'*}_p(q)[1-F^*_p(q)] $ 

$= \int_{M}^{x_1} g_p(q,x)f(x)dx \cdot \int_{M}^{\infty} g_p(q,x)f(x)dx$

$+ (\frac{c_2}{c_1}-1) \int_{M}^{x_1} g_p(q,x)f(x)dx \cdot \int_{x_1}^{\infty} g_p(q,x)f(x)dx$

$+ (\frac{c_1}{c_2}-1) \int_{x_1}^{\infty} g_p(q,x)f(x)dx \cdot \int_{M}^{x_1} g_p(q,x)f(x)dx$

$+ \int_{x_1}^{\infty} g_p(q,x)f(x)dx \cdot \int_{M}^{\infty} g_p(q,x)f(x)dx$

$-a $\\

$-f^{'*}_p(q)[1-F^*_p(q)] $ 

$= \int_{M}^{\infty} g_p(q,x)f(x)dx \cdot \int_{M}^{\infty} g_p(q,x)f(x)dx$

$+ (\frac{c_2}{c_1} + \frac{c_1}{c_2}-2) \int_{M}^{x_1} g_p(q,x)f(x)dx \cdot \int_{x_1}^{\infty} g_p(q,x)f(x)dx - a$

$= f^*_p(q)^2 
+ (b \int_{M}^{\infty} g_p(q,x)f(x)dx \cdot \int_{M}^{\infty} g_p(q,x)f(x)dx - a) $ \ \ \ \ \ (with $b= \frac{c_2}{c_1} + \frac{c_1}{c_2}-2>1$, because $\frac{c_2}{c_1} > 1$)\\

Note that $b$ can be made as large as possible for fixed $a$,
and so we get the following:

$-f^{'*}_p(q)[1-F^*_p(q)] $ 
$> \int_{M}^{\infty} g_p(q,x)f(x)dx \cdot \int_{M}^{\infty} g_p(q,x)f(x)dx$ \ \ 
$ = f^*_p(q)^2$ \\

Thus, for any MHR distribution $F$ defined as above, \textbf{$F_p$* is anti-MHR}.\\

We now provide an example distribution ($F_{p,(\varepsilon, k)}$) that satisfies the conditions of the distribution class characterized by $F$. Indeed, the distribution $F_{p,(\varepsilon, k)}$ is just a special case of the distribution $F$. By setting $c_1 = \varepsilon$, $c_2 = 1$, and $F(x_1) = \frac{1}{k}$ for the distribution $F$, we obtain the distribution $F_{p,(\varepsilon, k)}$. 
Thus, MHR $F$ can, in fact, lead to non-MHR (and even lead to anti-MHR) $F^*_{p}$. In other words, non-decreasing $h$ can lead to non-monotone (and even lead to non-increasing) $h^*_{p}$ because the MHR condition $-f^{'*}_p(q)[1-F^*_p(q)] \leq f^*_p(q)^2$ is violated, as shown in the example above where $F = F_{\varepsilon,k}$ and $F^*_{p} = F^*_{p,(\varepsilon,k)}$.\\

Therefore, there exists $\varepsilon\in (0,1), k > 1$, \ such that $F_{\varepsilon,k}$ is MHR but $F^*_{p,(\varepsilon, k)}$ is \textbf{anti-MHR}.\\
\end{proof}


We now define a setting where the consumer is drawn from $D=F_{\varepsilon, k}^n$ where $F_{\varepsilon, k}$ is the distribution we defined above.
Assume providers $1,\ldots, n-1$ set price $p$, and provider $n$ sets price $q$. The consumer has $n-1$ external options valued at $\max_{i\neq n} v_i - p$ and will only purchase from the lone provider ($n$) at price $q$ if the utility from provider $n$ is better than the utility from  providers $1,\ldots, n-1$.\\
    
   Provider $n$ has payoff $q \cdot \Pr_{\vec{v} \leftarrow D}[n = \arg\max_i \{v_i - p_i\}]$ with $p_i = p, \forall j\neq n$ and $p_n = q$ and sets the revenue maximizing price arg$\max_q (q \int f_{\varepsilon, k}(x)\prod_{i\neq n} Pr[v_i-p<x-q])$
    
    $=$ arg$\max_q (q \int f_{\varepsilon, k}(x)\prod_{i\neq n} Pr[v_i<-q+x+p]dx)$
    
    $=$ arg$\max_q (q \int f_{\varepsilon, k}(x)\prod_{i\neq n}F_{\varepsilon, k}(-q+x+p)dx)$
    
    $=$ arg$\max_q (q \int f_{\varepsilon, k}(x)F_{\varepsilon, k}^{n-1}(-q+x+p)dx)$\\
    
    We now restrict our attention to the case where there are only two providers ($n=2$). Consumer valuations are drawn from  $D= F^2_{\varepsilon, k}$. Thus, provider n (i.e. $2$) sets the revenue maximizing price arg$\max_q (q \int f_{\varepsilon, k}(x)F_{\varepsilon, k}^{n-1}(-q+x+p)dx)$ = arg$\max_q (q \int f_{\varepsilon, k}(x)F_{\varepsilon, k}(-q+x+p)dx)$\\
    
    Consider the case where $q \geq p$. Whenever $x < q-p$, provider $n$ will never be chosen since $Pr[max_{i\neq n}v_i-p < x-q]=0$, so we get the following integral bounds:
    arg$\max_q [q\cdot \int_{q-p}^{\infty} f_{\varepsilon, k}(x)F_{\varepsilon, k}(-q+x+p)dx]$
    = arg$\max_q$ \ \ $q \cdot  \
    ( \frac{1}{2} e^{\varepsilon (-q+p)}
    - \frac{1}{k^2} [\frac{1}{1+\varepsilon} - \frac{1}{2}] [e^{-\varepsilon (-q+p)} - e^{(-q+p)}]$ )\\
    
    We can now conclude that, in the two-provider scenario, for all $\varepsilon \in (0,1), k \geq 1$, when provider $1$ uses the price $p$, provider $2$ best responds by setting the revenue maximizing price arg$\max_q \ \ q \cdot 
    (\frac{1}{2} e^{\varepsilon (-q+p)}
    - \frac{1}{k^2} [\frac{1}{1+\varepsilon} - \frac{1}{2}] [e^{-\varepsilon (-q+p)} - e^{(-q+p)}]$ )\\
    
    We now use specific instances of this best-response result to prove the lemmas below.\\
    

\begin{lem}
When $\varepsilon=0.1$, $k=2$, $n=2$, the MHR distribution $D = F_{\varepsilon,k}^n = F_{0.1,2}^2$ has no symmetric equilibrium. 
    \end{lem}
    
    \begin{proof}
    When there are only two providers ($n=2$), our setting then becomes that of a consumer drawn from  $D= F^2_{\varepsilon, k}$. In particular for $k=2$, i.e. when the consumer is drawn from  $D= F^2_{\varepsilon, 2}$, provider $2$ is best responding by setting the revenue maximizing price arg$\max_q \ \ q \cdot
    (\frac{1}{2} e^{\varepsilon (-q+p)}
    - \frac{1}{k^2} [\frac{1}{1+\varepsilon} - \frac{1}{2}] [e^{-\varepsilon (-q+p)} - e^{(-q+p)}]$ ) = arg$\max_q \ \ q \cdot  ( \frac{1}{2} e^{\varepsilon (-q+p)}
    - \frac{1}{4} [\frac{1}{1+\varepsilon} - \frac{1}{2}] [e^{-\varepsilon (-q+p)} - e^{(-q+p)}]$ )\\

    Using the formula from Proposition $3.2$, we compute the hazard rate of the second highest value
    $h_2^n(F_{\varepsilon, k}) = \frac{3}{4} \cdot \varepsilon + \frac{1}{4} = \frac{3\varepsilon+1}{4} = \frac{3\cdot 0.1 + 1}{4} = \frac{13}{40}$ which gives the unique potential symmetric equilibrium price
    $p = 1/h_2^n(F_{\varepsilon, k}) = \frac{40}{13}$\\

    Provider $2$ is now best responding to this potential symmetric equilibrium price
    $p = \frac{40}{13}$ by setting the revenue maximizing price defined above as arg$\max_q \ \ q \cdot  ( \frac{1}{2} e^{\varepsilon (-q+\frac{40}{13})}
    - \frac{1}{4} [\frac{1}{1+\varepsilon} - \frac{1}{2}] [e^{-\varepsilon (-q+\frac{40}{13})} - e^{(-q+\frac{40}{13})}]$ ). Solving for this gives $q=3.50618 \neq p$, at a revenue of $1.53855$.
    
    The prices are not equal, thus we don't have \textbf{symmetric equilibrium}\\
    
    Therefore, when $\varepsilon=0.1$, $k=2$, $n=2$, the MHR distribution $D = F_{\varepsilon,k}^n$ has no symmetric equilibrium. \\
\end{proof}


\begin{observation}If $F$ is MHR, then $H_{2}^{n}(F)$ is weakly decreasing in $n$. If $F$ is anti-MHR, then $H_{2}^{n}(F)$ is weakly increasing in $n$.
\end{observation}


\begin{lem}
When $n=2$, $\varepsilon > \frac{1}{27}$, and $k = 2$, $D = F^n_{\varepsilon,k}$ satisfies the Limit Entry Condition, but does not have a symmetric equilibrium in the Free Market setting.
\end{lem}

\begin{proof}
    
    We refer to the same example above (in Lemma $A.2$) where we had $n=2$, $\varepsilon = 0.1$, and $k = 2$. We showed that it has \textbf{no symmetrical equilibrium}. However, any $\varepsilon > \frac{1}{27}$ guarantees that $H_1^n(F_{\varepsilon, k}) < \frac{n}{h_2^n(F_{\varepsilon, k})}$ and thus satisfies the \textbf{Limit-Entry Condition}. In particular, when $\varepsilon = 0.1 > \frac{1}{27}$, $D = F^2_{\varepsilon,2}$ satisfies the Limit Entry Condition but does not have a symmetric equilibrium in the Free Market setting.\\
    
    Therefore, when $n=2$, $\varepsilon  > \frac{1}{27}$, and $k = 2$, $D = F^n_{\varepsilon,k}$ satisfies the Limit Entry Condition, but does not have a symmetric equilibrium in the Free Market setting.\\
    
\end{proof}


\begin{lem} 
When $n=2$, $\varepsilon = 0.02$, and $k = \frac{4}{3}$, $D=F^n_{\varepsilon, k}$ has a symmetric equilibrium in the Free Market setting, but does not satisfy the Limit Entry Condition. 
\end{lem}
\begin{proof}
We again consider the two-provider setting ($n=2$) where
the consumer is drawn from  $D= F^2_{\varepsilon, k}$. In this instance, we set $\varepsilon = 0.02$, and $k = \frac{4}{3}$, i.e. when the consumer is drawn from  $D= F^2_{0.02, \frac{4}{3}}$. 
Using the formula from Proposition $3.2$, we compute the hazard rate of the second highest value
    $h_2^n(F_{\varepsilon, k}) = \frac{13}{160}$ which gives the unique potential symmetric equilibrium price
    $p = 1/h_2^n(F_{\varepsilon, k}) = \frac{160}{13} = 12.308$\\
    
    Provider $2$ is now best responding to this potential symmetric equilibrium price
    $p =  \frac{160}{13}$ by setting the revenue maximizing price defined above as 
     arg$\max_q \ \ q \cdot  \ (
    \frac{1}{2} e^{0.02(-q+\frac{160}{3})}
    - (\frac{3}{4})^2 [\frac{1}{1+0.02} - \frac{1}{2}] [e^{-0.02 (-q+\frac{160}{3})} - e^{(-q+\frac{160}{3})}]$ )
    for $q \geq \frac{160}{3}$. Solving for this gives $q =\frac{160}{3} =p$ and revenue = $\frac{80}{3}$.

    We thus have a \textbf{symmetrical equilibrium} since the best response to $p$ is $q=p$.\\
    
    We now check if the Limit Entry condition is satisfied. First, we compute the terms for the Limit-Entry Condition as follows: 
    
    $H_1^n(F_{\varepsilon, k}) =  (\frac{3}{4})^2\cdot\frac{1}{0.02} + [1-(\frac{3}{4})^2]\cdot 1 = 28.5625$ (exact)
    
    $h_2^n(F_{\varepsilon, k}) = [1 - (1-\frac{3}{4})^2]\cdot 0.02 + (1-\frac{3}{4})^2\cdot 1 =\frac{15}{16}\cdot 0.02 + \frac{1}{16}=\frac{13}{160}$ which gives $\frac{1}{h_2^n(F_{\varepsilon, k})} = \frac{160}{13} = 12.308$
    
    Note that $H_1^n(F_{\varepsilon, k}) > \frac{2}{h_2^n(F_{\varepsilon, k})}$ which violates the \textbf{Limit-Entry Condition}.\\

    Therefore, when $n=2$, $\varepsilon = 0.02$, and $k = \frac{4}{3}$, $D=F^n_{\varepsilon, k}$ has a symmetric equilibrium in the Free Market setting, but does not satisfy the Limit Entry Condition. \\
\end{proof}


\end{document}